\documentclass{llncs}

\usepackage{epic}
\usepackage{times}
\usepackage{xspace}
\usepackage{graphicx}
\usepackage{latexsym}
\usepackage{amsmath}
\usepackage{amssymb}
\usepackage{algorithm}
\usepackage[noend]{algpseudocode}
\usepackage{hyperref}
\usepackage{comment}
\usepackage{paralist}
\usepackage{array}
\usepackage{rotating}

\usepackage[utf8]{inputenc}
\usepackage[T1]{fontenc}

\usepackage{tikz}
\usetikzlibrary{positioning,arrows}

\newcommand{\mymod}{\mbox{\scriptsize \rm \ mod \ }}

\newcommand{\myomit}[1]{}

\newcommand{\set}[1]{\ensuremath{\left\{#1\right\}}}
\newcommand{\seq}[1]{\ensuremath{\left[#1\right]}}
\newcommand{\setcomp}[2]{\ensuremath{\{#1\ |\ #2\}}}

\newcommand{\funid}[1]{\ensuremath{#1}}
\newcommand{\algid}[1]{\mbox{\sc #1}}

\newcommand{\zero}{\ensuremath{0}}

\newcommand{\DC}{\emph{DC}}
\newcommand{\SAT}{SAT}
\newcommand{\CNF}{CNF}

\newcommand{\var}{\funid{var}}
\newcommand{\lit}{\funid{lit}}

\newcommand{\SPACING}{\algid{Spacing}}
\newcommand{\SPACINGF}{\algid{Spacing}\ensuremath{_{F}}}
\newcommand{\SPACINGH}{\algid{Spacing}\ensuremath{_{h}}}
\newcommand{\SPACINGONE}{\algid{Spacing1}}
\newcommand{\BLINDSPACING}{\algid{Spacing}\ensuremath{_{SB}}}
\newcommand{\occ}{\funid{occ}}
\newcommand{\AllDifferent}{\algid{AllDifferent}}
\newcommand{\oModel}{\algid{$OM$}}
\newcommand{\sModel}{\algid{$SM$}}
\newcommand{\srModel}{\algid{$SR$}}
\newcommand{\bsModel}{\algid{$SB$}}

\newcommand{\vs}{\ensuremath{v}}
\newcommand{\cs}{\ensuremath{c}}

{\makeatletter
 \gdef\TODOmark{%
   \expandafter\ifx\csname @mpargs\endcsname\relax % in minipage?
     \expandafter\ifx\csname @captype\endcsname\relax % in figure/caption?
       \marginpar{TODO}% not in a caption or minipage, can use marginpar
     \else
       TODO % notice trailing space
     \fi
   \else
     TODO % notice trailing space
   \fi}
 \gdef\TODO{\@ifnextchar[\TODO@lab\TODO@nolab}
 \long\gdef\TODO@lab[#1]#2{{\bf [\TODOmark #2 ---{\sc #1}]}}
 \long\gdef\TODO@nolab#1{{\bf [\TODOmark #1]}}
}
\title{Global SPACING Constraint\\
(Technical Report)\thanks{NICTA is funded by
the Australian Government as represented by
the Department of Broadband, Communications and the Digital Economy and
the Australian Research Council. }}

\author{Nina Narodytska\inst{1} \and
Peter Skočovský\inst{2} \and
Toby Walsh\inst{1}}
\institute{NICTA and UNSW, Sydney, Australia \and
Universidade Nova de Lisboa, Lisbon, Portugal}

\begin{document}

\maketitle

\begin{abstract}
We propose a new global $\SPACING$ constraint
that is useful in modeling events 
that are distributed over time, like 
learning units scheduled over a study program or repeated patterns 
in music compositions. 
First, we investigate theoretical properties of the constraint
and identify tractable special cases. We propose efficient $\DC$ filtering
algorithms for these cases. Then, we experimentally
evaluate performance of the proposed algorithms on a music composition
problem and demonstrate that our filtering algorithms outperform the 
state-of-the-art approach for solving this problem.     
\end{abstract}

\section{Introduction}

When studying a new topic,
it is often better to spread the learning over a long period of time
and to revise topics repeatedly.
This ``spacing effect'' was first identified by Hermann Ebbinghaus in 1885 \cite{spacingeffect}.
It has subsequently become ``one of the most studied phenomena in the 100-year history of learning research'' \cite{dempster88}.
It has been observed across domains
(e.g. learning mathematical concepts or a foreign language, as well as learning a motor skill),
across species (e.g. in pigeons, rats and humans),
across age groups and individuals, and across timescales (e.g. from seconds to months).
To enable learning software to exploit this effect,
Novikoff, Kleinberg and Strogratz have proposed a simple mathematical model \cite{nkspnas12}.
They consider learning a sequence of educational units,
and model the spacing effect with a constraint defined by two sequences,
$A = \seq{a_1, a_2, \ldots}$ and $B = \seq{b_1, b_2, \ldots}$.
For each unit being taught, the $i + 1$st time that it should be reviewed
is between $a_i$ and $b_i$ time steps after the $i$th time.

This technical report, after giving a brief background in the following section,
describes and analyses the Global \SPACING\ Constraint in the section 3.
Then the sections 4, 5, 6, 7 and 8 explore restrictions of the constraint
and identify tractable special cases. In the sections 5 and 6 we propose efficient
filtering algorithms for these tractable cases.
In the section 9, we describe a useful application of the
constraint to solving a music composition problem.
The section 10 presents experimental evaluation of the \SPACING\ constraint
on the music composition problems.
%  in the search for tractability
% and describe propagators for some of them.
% One of the tractable restrictions is used in a number of models proposed in the section 9
% and the experimental results are analysed in the section 10.

\section{Background}

\subsection{Constraint Satisfaction Problems (CSP)}

A constraint satisfaction problem consists of a set of
variables, each with a finite domain of values, and a set of
constraints specifying allowed combinations of values for
subsets of variables. We use capital letters for variables 
(e.g. $X$ and $Y$), and lower case for values (e.g. $d$ and $d'$).
We write $D(X)$ for the domain of a variable $X$ and
$D = \bigcup_{i = 1}^{n} D(X_i)$ for the set of all the domain values.
Assigning a value $d \in D(X)$ to a variable $X$ means removing all the other values from its domain.
A solution is an assignment of values to the variables satisfying the constraints.
Constraint solvers typically explore partial assignments enforcing
a local consistency property using either specialized or general
purpose propagation (or filtering) algorithms.
% A \emph{support} for a constraint $C$ is a tuple %$\tau$ %on $var(C)$
% that assigns a value to each variable from its domain which satisfies $C$.
A \emph{support} for a constraint $C$ is an assignment that
assigns to each variable some value from its domain and satisfies $C$.
A constraint is \emph{domain consistent} (\DC) iff for each
variable $X_i$, every value in $D(X_i)$ belongs to some support.
% \begin{proposition}\label{prop:harddc}
% If finding a support of a constraint is NP-hard,
% then also establishing \DC\ on the constraint is NP-hard.
% \end{proposition}
% \begin{proof}
% Having a procedure that finds a support for a particular constraint,
% we can establish \DC\ on the constraint by assigning to each variable each value from its domain,
% executing the procedure and removing the value if no support was found.
% \qed
% \end{proof}

\subsection{Matching Theory}

We also give some background on matching.
A \emph{bipartite graph} is a graph $G=(U,V,E)$ with
the set of nodes partitioned between $U$ and $V$ such that
there is no edge between two nodes in the same partition.
A \emph{matching} in a graph $G$ is a subset of $E$ where no two edges have a node in common.
A \emph{maximum matching} is a matching of maximum cardinality.
Regin\cite{regin1} proposed efficient propagator based on a maximum matching algorithm.

\subsection{Propositional Satisfiability (SAT)}

A \emph{Propositional Satisfiability} (\SAT) is a problem of
finding a model of a propositional formula.
The propositional formula is usually in \emph{Conjunctive Normal Form} (\CNF),
which is a set of clauses. We consider the clauses to be sets of literals.
A literal is either a propositional variable (i.e. $p$) or
a negated propositional variable (i.e. $\neg p$).
The set of all variables from $\phi$ is $\var(\phi)$ and
$\neg \var(\phi) = \setcomp{\neg p}{p \in \var(\phi)}$.
Without loss of generality we can order the variables from $\var(\phi)$ into sequence $\seq{p_1, \ldots, p_\vs}$
and the clauses of $\phi$ into sequence $\seq{C_1, \ldots, C_\cs}$.
Let $\lit(\phi) = \var(\phi) \cup \neg\var(\phi)$.
Then a set of literals $I \subseteq \lit(\phi)$ is an interpretation of $\phi$ if it is maximal and consistent,
i.e. it does not contain a pair of complementary literals $l \in I \rightarrow \overline{l} \notin I$.
An interpretation $I$ of $\phi$ is a model of $\phi$ if
it contains at least one literal from each clause of $\phi$.
We also use the following notation: if $L$ is a set of literals, then
$L' = \setcomp{l'}{l \in L}$ and $L^i = \setcomp{l^i}{l \in L}$.
% We define bijection $\litid$ between literals and integers such that
% $\litid(p_i) = i$ and $\litid(\neg p_i) = -i$.

\section{The Global \SPACING\ Constraint}

First we describe the Global \SPACING\ Constraint
and then its further restrictions in the following sections.
For simplicity, we introduce a function that
returns the number of occurrences of a domain value $d \in D$
in a sequence of variables $X = \seq{X_1, \ldots, X_n}$.
It is $\occ(d, X) = |\setcomp{i}{X_i = d, 1 \leq i \leq n}|$
Now we define the Global \SPACING\ Constraint as follows:

\begin{definition}\label{def:spacing}
Let $X = \seq{X_1, \ldots, X_n}$ be a sequence of $n$ variables and
let $S \subseteq D$ be a set of domain values.
Let $A = \seq{a_1, \ldots, a_{k-1}}$ and $B = \seq{b_1, \ldots, b_{k-1}}$ be sequences of natural numbers such that
$a_i \leq b_i$ for $1 \leq i \leq k-1$.
Then $\SPACING(S, A, B, X)$ holds iff
for all $i$ s.t. $1 \leq i \leq k-1$ and for all $d \in S$ it holds that
if there exists $j$ s.t. $1 \leq j \leq n$ and $X_j = d$ and $\occ(d, \seq{X_1, \ldots, X_j}) = i$ then
there exists $j' \leq n$ s.t. $j + a_i \leq j' \leq j + b_i$ and $X_{j'} = d$ and $\occ(d, \seq{X_1, \ldots, X_{j'}}) = i+1$.
\end{definition}

In other words, each value $d \in S$ either does not occur in $X$ at all,
or it occurs on at least $k$ different places and
the distances between the places are determined by the sequences $A$ and $B$.
% its occurrences have to comply with \emph{minimum} and \emph{maximum distance conditions}
% determined by the sequences $A$ and $B$ respectively.
The \emph{minimum distance condition} forces the $i+1$st occurrence of the value $d$
to be no closer than $a_i$ places from its $i$th occurrence and
the \emph{maximum distance condition} forces the $i+1$st occurrence
to be no further than $b_i$ places from the $i$th occurrence.

For example, suppose we need to prepare a playlist for a radio station.
The \SPACING\ Constraint allows us to specify that any of top ten songs is
either not played at all, or it is played at least four times in the 360 song long playlist and
it is not repeated more frequently than every 30 songs, but at least every 90 songs.
The constraint would be imposed on the sequence $X = \seq{X_1, \ldots, X_{360}}$ of
$n = 360$ variables. The domain values would represent the songs and $S \subseteq D$
would be the set of the top ten songs. The number of spaced occurrences is $k = 4$,
so the sequences $A$ and $B$ would be of length $k-1 = 3$.
Overall the constraint would be specified as $\SPACING(S, \seq{30,30,30}, \seq{90,90,90}, X)$.

\begin{theorem}\label{thm:intractS}
Enforcing \DC\ on the Global \SPACING\ Constraint is NP-hard.
\end{theorem}

\begin{proof}
We prove this by reduction of \SAT\ to the problem of finding a support for \SPACING.
Let $\phi$ be an arbitrary \CNF\ with $\vs$ propositional variables and $\cs$ clauses.
We will abuse the notation slightly by using literals as domain values.
There is a model of $\phi$ iff there is a support for
the constraint $\SPACING(S, \seq{a_1, \ldots, a_{k-1}}, \seq{b_1, \ldots, b_{k-1}}, X)$ with
\begin{itemize}
	\item $S = \lit(\phi)$
	\item $k = \cs+1$
	\item $a_i = 1$ and $b_i = \vs+1$ for $1 \leq i \leq k-1$
	\item $X$ is a sequence of $\vs\cs+\vs+\cs$ variables with domains as described below.
\end{itemize}
If we cut the sequence $X$ into slices $\vs+1$ variables long and put the slices under each other,
we obtain a table with $\cs+1$ rows where the last cell is empty.
For simplicity the variables will be indexed $X_{j,i}$, where
$j$ is a row number and $i$ is a column number,
i.e. $X_{j,i}$ stands for $X_{(j-1)(\vs+1) + i}$.
The first $\vs$ columns will represent the propositional variables of $\phi$, so
\begin{itemize}
	\item the domains of variables $X_{j,i}$ for $1 \leq i \leq \vs$ and $1 \leq j \leq \cs+1$ are $\set{p_i, \neg p_i}$.
\end{itemize}
The last column represents $\phi$, each clause in one row, so
the domains are sets of literals from particular clauses. Simply,
\begin{itemize}
	\item the domains of variables $X_{j,\vs+1}$ for $1 \leq j \leq \cs$ are $C_j$.
\end{itemize}

For example take CNF $\phi = \seq{\set{\neg p, q, r}, \set{\neg q, r}, \set{\neg p, \neg q}, \set{p, q}}$
the equivalent constraint would be $\SPACING(S, \seq{1,1,1,1}, \seq{4,4,4,4}, X)$ with
$S = \set{p, q, r, \neg p, \neg q, \neg r}$ and $X$ would contain $19$ variables with
domains ordered into the following table:
\begin{center}
\begin{tabular}{r|cccc}
$i = $ & 1 & 2 & 3 & 4 \\
\hline
$j = 1$ & $\set{p, \neg p}$ & $\set{q, \neg q}$ & $\set{r, \neg r}$ & $\set{\neg p, q, r}$ \\ 
$j = 2$ & $\set{p, \neg p}$ & $\set{q, \neg q}$ & $\set{r, \neg r}$ & $\set{\neg q, r}$ \\ 
$j = 3$ & $\set{p, \neg p}$ & $\set{q, \neg q}$ & $\set{r, \neg r}$ & $\set{\neg p, \neg q}$ \\ 
$j = 4$ & $\set{p, \neg p}$ & $\set{q, \neg q}$ & $\set{r, \neg r}$ & $\set{p, q}$ \\ 
$j = 5$ & $\set{p, \neg p}$ & $\set{q, \neg q}$ & $\set{r, \neg r}$ & \\ 
\end{tabular}
\end{center}

Note that $D \subseteq S$.
So if a value occurs in a support for the constraint, it has to occur in it at least $\cs+1$ times.
Due to construction of the domains, one value can occur on at most $2\cs+1$ places
(in one of the first $\vs$ columns and in the last column).
Each value is sharing $\cs+1$ of these places with its complement,
thus if a value occurs in a support, it occupies at least $\cs+1$ of $2\cs+1$ places,
so its complement does not have space to repeat enough times to satisfy the constraint.
Hence if a value occurs in a support it will occupy one of the first $\vs$ columns completely,
because its complement cannot.

Consequently, if we have a support for the constraint,
the set of values assigned to the first $\vs$ variables is a model of $\phi$,
because each value selected in the last column has to occur in it and
complement of none of these values can occur in it.

On the other hand, if we have a model of $\phi$, we can obtain a support for the constraint by
assigning the literals from the model to all the variables as a values.
This is always possible, because a model interprets all the literals and it is consistent,
so there will be exactly one value for each of the variables in the first $\vs$ columns, and
a model satisfies each clause, so there will be some value for each variable in the last column.
This will really be a support for the constraint, because
each value that will occur in it will occur in each row at least once.
% 
% While we can obtain a model of $\phi$ from a support of the constraint and vice versa,
% by \TODO{proposition 1} %TODO
% enforcing \DC\ on the Global \SPACING\ Constraint is NP-hard.
\qed
\end{proof}

Please, note that the proof does not make use of the condition of minimal spacing
imposed by the sequence $A$.
Also the sequence $B$ is constant ($b_i = b_{i+1}$ for all $1 \leq i \leq k-2$),
so the proof holds also for much simpler constraint.
In fact, the proof does not use the condition that the values have to be repeated,
it is based barely on the number of occurrences.
The reason why \SPACING\ is NP-hard is that there is a possibility for the values to
either occur in the support or not.
Next we will examine a case that does not provide such a possibility.

\section{The Forced Global \SPACING\ Constraint (\SPACINGF)}

This section describes a restriction of the constraint where
all the values from $S$ are forced to occur in the sequence and
later it analyses also restriction with the distance conditions relaxed.

\begin{definition}\label{def:spacingF}
Let $X = \seq{X_1, \ldots, X_n}$ be a sequence of $n$ variables and
let $S \subseteq D$ be a set of domain values.
Let $A = \seq{a_1, \ldots, a_{k-1}}$ and $B = \seq{b_1, \ldots, b_{k-1}}$ be sequences of natural numbers such that
$a_i \leq b_i$ for $1 \leq i \leq k-1$.
Then $\SPACINGF(S, A, B, X)$ holds iff
$\SPACING(S, A, B, X)$ holds and all $S \subseteq \set{X_1, \ldots, X_n}$.
\end{definition}
This means that each value $v \in S$ occurs in the sequence on at least $k$ different places.

\begin{theorem}\label{thm:intractSF}
Enforcing \DC\ on the Global \SPACINGF\ Constraint is NP-hard.
\end{theorem}

\begin{proof}[no minimal distance condition]
We prove this by reduction of \SAT\ to the problem of finding a support for \SPACINGF.
The reduction is the same as in the proof of Theorem~\ref{thm:intractS}, except
the sequence $X$ has additional $(\cs+2)(\vs+1)+1$ variables (in total $(2\cs+3)(\vs+1)$ variables)
with domains as follows:
If we order the variables into the same table as before, with rows of length $\vs+1$,
the variables in the rest of the last column, as well as the variables in the $\cs+2$nd row
can only take a dummy value \zero
\begin{itemize}
	\item the domains of variables $X_{j,\vs+1}$ for $\cs+1 \leq j \leq 2\cs+3$ and
	variables $X_{\cs+2,i}$ for $1 \leq i \leq \vs$ are $\set{\zero}$,
	where $\zero \not\in S$
\end{itemize}
and the first $\vs$ cells of remaining $\cs+1$ rows are
again the representation of propositional variables, so
\begin{itemize}
	\item the domains of variables $X_{j,i}$ for $1 \leq i \leq \vs$ and
	$\cs+3 \leq j \leq 2\cs+3$ are again $\set{p_i, \neg p_i}$.
\end{itemize}

For example, the CNF $\phi = \seq{\set{\neg p, q, r}, \set{\neg q, r}, \set{\neg p, \neg q}, \set{p, q}}$
from the previous example would be reduced to the constraint
$\SPACINGF(S, \seq{1,1,1,1}, \seq{4,4,4,4}, X)$ with
$S = \set{p, q, r, \neg p, \neg q, \neg r}$ and $X$ would contain $44$ variables with
domains ordered into the following table:
\begin{center}
\begin{tabular}{r|cccc}
$i = $ & 1 & 2 & 3 & 4 \\
\hline
 $j = 1$ & $\set{p, \neg p}$ & $\set{q, \neg q}$ & $\set{r, \neg r}$ & $\set{\neg p, q, r}$ \\ 
 $j = 2$ & $\set{p, \neg p}$ & $\set{q, \neg q}$ & $\set{r, \neg r}$ & $\set{\neg q, r}$ \\ 
 $j = 3$ & $\set{p, \neg p}$ & $\set{q, \neg q}$ & $\set{r, \neg r}$ & $\set{\neg p, \neg q}$ \\ 
 $j = 4$ & $\set{p, \neg p}$ & $\set{q, \neg q}$ & $\set{r, \neg r}$ & $\set{p, q}$ \\ 
 $j = 5$ & $\set{p, \neg p}$ & $\set{q, \neg q}$ & $\set{r, \neg r}$ & $\set{\zero}$ \\ 
 $j = 6$ & $\set{\zero}$ & $\set{\zero}$ & $\set{\zero}$ & $\set{\zero}$ \\ 
 $j = 7$ & $\set{p, \neg p}$ & $\set{q, \neg q}$ & $\set{r, \neg r}$ & $\set{\zero}$ \\ 
 $j = 8$ & $\set{p, \neg p}$ & $\set{q, \neg q}$ & $\set{r, \neg r}$ & $\set{\zero}$ \\ 
 $j = 9$ & $\set{p, \neg p}$ & $\set{q, \neg q}$ & $\set{r, \neg r}$ & $\set{\zero}$ \\ 
$j = 10$ & $\set{p, \neg p}$ & $\set{q, \neg q}$ & $\set{r, \neg r}$ & $\set{\zero}$ \\ 
$j = 11$ & $\set{p, \neg p}$ & $\set{q, \neg q}$ & $\set{r, \neg r}$ & $\set{\zero}$ \\ 
\end{tabular}
\end{center}

We can divide the sequence into 2 parts separated by the $\cs+2$nd row.
We will call the first one \emph{positive} part and the second one \emph{negative} part.
If a value from $S$ occurs in one part, it will occur at least $\cs+1$ times in that part,
because the gap of \zero-s between the parts is longer than any $b_i$, so
all the necessary repetitions must occur in one part.
Also if a value $d \in S$ occurs in one part, its complement $\overline{d}$ cannot occur in the same part
because of a counting argument analogous to the one in the proof of Theorem~\ref{thm:intractS}, so,
while all the values from $S$ must occur in the support,
the complement $\overline{d}$ must occur in the other part (where $d$ will not occur due to the same argument).

Thanks to this, the proof of Theorem~\ref{thm:intractS} holds for the positive part and
the values that cannot occur in the positive part will occur in the negative part.
It is obvious that there is enough space for all the values:
$|S| = 2\vs$ while $\vs$ values occur in the positive part and the other $\vs$ values occur in the negative part.
Hence, we can obtain a support for the constraint from a model of $\phi$
and a model from a support in the same manner as in the proof of Theorem~\ref{thm:intractS}.
\qed
\end{proof}

As before, this proof does not make use of the condition on the minimal distance between
two consecutive values from $S$ ($A$ is a sequence of 1-s).
In fact, also if we relax the condition on the maximal distance between
two consecutive values from $S$ ($B$ will be a sequence of $n$-s),
the constraint is still intractable.

\begin{theorem}\label{thm:intractB}
Enforcing \DC\ on the Global \SPACINGF\ Constraint with $b_i = n$ for $1 \leq i \leq k-1$ is NP-hard.
\end{theorem}

\begin{proof}[no maximal distance condition]
We prove this by reduction of \SAT\ to the problem of finding a support for \SPACING.
Let $\phi$ be an arbitrary \CNF\ with $\vs$ propositional variables and $\cs$ clauses.
There is a model of $\phi$ iff there is a support for
the constraint $\SPACINGF(S, \seq{a_1, \ldots, a_{k-1}}, \seq{b_1, \ldots, b_{k-1}}, X)$ with
\begin{itemize}
	\item $S = \lit(\phi)$
	\item $k = \cs$
% 	\item $a_i = 5\vs+1$ and $b_i = 9\vs+1$ for $1 \leq i \leq k-1$
	\item $a_i = 5\vs+1$ and $b_i = (7\vs+1)\cs$ for $1 \leq i \leq k-1$
	\item $X$ is a sequence of $(7\vs+1)\cs$ variables with the domains as described below.
\end{itemize}
If we organize the variables from $X$ into a table with $\cs$ rows and $7\vs+1$ columns,
the first column will represent $\phi$ in the way that
\begin{itemize}
	\item the domains of variables $X_{j,1}$ for $1 \leq j \leq \cs$ are $C_j$,
\end{itemize}
following $2\vs$ columns together with the first column will represent satisfied literals
\begin{itemize}
	\item the domains of variables $X_{j, i+1}$ for $1 \leq i \leq \vs$ and $1 \leq j \leq \cs$ are $\set{p_i, \zero}$ and
	\item the domains of variables $X_{j, i+(1+\vs)}$ for $1 \leq i \leq \vs$ and $1 \leq j \leq \cs$ are $\set{\neg p_i, \zero}$,
\end{itemize}
following $2\vs$ columns will be a padding of \zero-s
\begin{itemize}
	\item the domains of variables $X_{j, i+(1+2\vs)}$ for $1 \leq i \leq 2\vs$ and $1 \leq j \leq \cs$ are $\set{\zero}$,
\end{itemize}
following $\vs$ columns will represent unsatisfied literals
\begin{itemize}
	\item the domains of variables $X_{j, i+(1+4\vs)}$ for $1 \leq i \leq \vs$ and $1 \leq j \leq \cs$ are $\set{p_i, \neg p_i}$
\end{itemize}
and the last $2\vs$ columns will be again a padding of \zero-s.
\begin{itemize}
	\item the domains of variables $X_{j, i+(1+5\vs)}$ for $1 \leq i \leq 2\vs$ and $1 \leq j \leq \cs$ are $\set{\zero}$.
\end{itemize}

For example, the CNF from our running example would be reduced to the constraint
$\SPACINGF(S, \seq{16,16,16}, \seq{88,88,88}, X)$ with
$S = \set{p, q, r, \neg p, \neg q, \neg r}$ and $X$ would contain $88$ variables with
domains ordered into the following table:
\begin{center}
\newcommand{\compd}[1]{\scriptsize\begin{tabular}{>{$}c<{$}} #1 \end{tabular}}
% \begin{tabular}{r|*{22}{|>{\centering\arraybackslash}p{4mm}}}
\begin{tabular}{r|*{22}{|>{\centering\arraybackslash}c}}
$i = $ & 1 & 2 & 3 & 4 & 5 & 6 & 7 & 8 & 9 & 10 & 11 & 12 & 13 & 14 & 15 & 16 & 17 & 18 & 19 & 20 & 21 & 22 \\
\hline
\hline
$j = 1$ & \compd{\neg p \\ q \\ r} & \compd{p \\ \zero} & \compd{q \\ \zero} & \compd{r \\ \zero} & \compd{\neg p \\ \zero} & \compd{\neg q \\ \zero} & \compd{\neg r \\ \zero} & \compd{\zero} & \compd{\zero} & \compd{\zero} & \compd{\zero} & \compd{\zero} & \compd{\zero} & \compd{p \\ \neg p} & \compd{q \\ \neg q} & \compd{r \\ \neg r} & \compd{\zero} & \compd{\zero} & \compd{\zero} & \compd{\zero} & \compd{\zero} & \compd{\zero} \\
\hline
$j = 2$ & \compd{\neg q \\ r} & \compd{p \\ \zero} & \compd{q \\ \zero} & \compd{r \\ \zero} & \compd{\neg p \\ \zero} & \compd{\neg q \\ \zero} & \compd{\neg r \\ \zero} & \compd{\zero} & \compd{\zero} & \compd{\zero} & \compd{\zero} & \compd{\zero} & \compd{\zero} & \compd{p \\ \neg p} & \compd{q \\ \neg q} & \compd{r \\ \neg r} & \compd{\zero} & \compd{\zero} & \compd{\zero} & \compd{\zero} & \compd{\zero} & \compd{\zero} \\
\hline
$j = 3$ & \compd{\neg p \\ \neg q} & \compd{p \\ \zero} & \compd{q \\ \zero} & \compd{r \\ \zero} & \compd{\neg p \\ \zero} & \compd{\neg q \\ \zero} & \compd{\neg r \\ \zero} & \compd{\zero} & \compd{\zero} & \compd{\zero} & \compd{\zero} & \compd{\zero} & \compd{\zero} & \compd{p \\ \neg p} & \compd{q \\ \neg q} & \compd{r \\ \neg r} & \compd{\zero} & \compd{\zero} & \compd{\zero} & \compd{\zero} & \compd{\zero} & \compd{\zero} \\
\hline
$j = 4$ & \compd{p \\ q} & \compd{p \\ \zero} & \compd{q \\ \zero} & \compd{r \\ \zero} & \compd{\neg p \\ \zero} & \compd{\neg q \\ \zero} & \compd{\neg r \\ \zero} & \compd{\zero} & \compd{\zero} & \compd{\zero} & \compd{\zero} & \compd{\zero} & \compd{\zero} & \compd{p \\ \neg p} & \compd{q \\ \neg q} & \compd{r \\ \neg r} & \compd{\zero} & \compd{\zero} & \compd{\zero} & \compd{\zero} & \compd{\zero} & \compd{\zero} \\
\end{tabular}
\end{center}

We can split the table into two parts.
We will address the first $2\vs+1$ columns as a \emph{positive} part and
the $4\vs + 2$nd to $5\vs + 1$-st column as a \emph{negative} part.
(The rest is padding of \zero-s.)
In order for an assignment to $X$ to be a support, the following must hold:
\begin{itemize}
	\item A value $d \in S$ can occur at most once per row because
	of the minimal distance condition imposed by the sequence $A$.
	While the number of rows is $k$, any value $d \in S$ has to occur in each row,
	so all values from $S$ have to occur in each row exactly once.
	\item Two complementary values $d, \overline{d} \in S$ cannot occur in the same part of any row.
	For the negative part this is due to construction of the domains and
	if both occurred in the positive part,
	none of them would be able to occur in the negative part of the same row,
	so there would be a variable with empty domain, which is a contradiction.
	\item If a value from $S$ occurs in one part of some row,
	it will occur in the same part in each row. This is because of the following:
	\begin{itemize}
		\item If a value $d \in S$ occurs in the negative part of a row number $j$,
		it cannot occur in the positive part of the following row $j+1$,
		because of the minimal distance condition,
		so it will occur in the negative part of the row $j+1$.
		\item Suppose a value $d \in S$ occurs in the positive part of a row number $j$.
		Then the complementary value $\overline{d}$ has to occur in the negative part of the row $j$. %,
% 		because $d$ can occur only once per row and the construction of the domains.
		Now if $j \neq k$ and if the value $d$ did not occur in the positive part of the following row $j+1$,
		then it would have to occur in the negative part.
		Then, however, the complementary value $\overline{d}$ would not be able to occur in
		that negative part of the row $j+1$, which is contradiction with the above.
		So $d$ has to occur in the positive part of the row $j+1$.
		\item Each value from $S$ has to occur in the first row.
		So if a value from $S$ occurred in some row in a different part as in the first row,
		we would reach contradiction by applying the above.
	\end{itemize}
\end{itemize}

Consequently, if we have a support of the constraint,
the assignment to the first $2\vs+1$ variables (the positive part) represents a model of $\phi$.
The set of the values assigned to this variables (without the value $\zero$) is clearly an interpretation,
because it cannot contain complementary literals.
Some value is selected for each variable representing clause (first column) and
this value has to occur in the positive part of each row and
its complement cannot occur in the positive part,
hence the interpretation is also a model, because
it contains at least one literal from each clause.

We can obtain a support of the constraint from a model $I$ of $\phi$ by
assigning a value that represents a literal that is satisfied by $I$ to the variables in the first column.
This is always possible, because $I$ must satisfy at least one literal from each clause.
The assignment to the variables in the positive part will be as follows:
\begin{inparaenum}[(1)]
	\item The value \zero\ is assigned to the variable with
	domain containing value already assigned to the first variable in the row.
	\item The literals satisfied by $I$ are assigned to the rest of the variables as values.
	\item Where this is not possible, \zero\ is assigned.
\end{inparaenum}
This is always possible, because each literal that may occur in $\phi$
is represented by some value from the domains of these variables and
no two representations of two different literals occur in a domain of one variable.
Further, we assign the variables in the negative part so that
they contain representations of each literal that is falsified by $I$.
Such an assignment is always possible because:
\begin{inparaenum}[(1)]
	\item Each of these variables will be assigned, because $I$ falsifies either $p_i$ or $\neg p_i$ for each $1 \leq i \leq \vs$.
	\item And there will be at most one possible value for each variable, because $I$ is consistent.
\end{inparaenum}
We assign \zero\ to the rest of the variables (their domains are $\set{\zero}$).
Now the assignment is a support of the constraint, because
each value $d \in S$ occurs exactly once in each row which is $\cs = k$ occurrences.
Further, due to the fact that
the representations of the satisfied literals occur always in the positive part and
the representations of the falsified literals occur always in the negative part,
there is always at least $5\vs$ places between two successive occurrences of each $d \in S$.
In details, the last possible occurrence of a representation of a satisfied literal in the $j$th row is on the position
$(j-1)(7\vs+1)+(2\vs+1)$ and
the next possible occurrence of a representation of a satisfied literal is at the beginning of the following row
$(j+1-1)(7\vs+1)+1$.
The difference of these two positions is $5\vs+1 \geq 5\vs$.
% The first possible occurence of a representation of a satisfied literal in the $j$-th row is on position
% $(j-1)(7\vs+1)+1$ and
% the last possible position where this representation will repeat is
% $(j+1-1)(7\vs+1)+(2\vs+1)$.
% The difference of these two positions is $9\vs+1 \leq 9\vs+1$.
% Analogous reasoning holds for the representations of unsatisfied literals.
The last possible occurrence of a representation of a falsified literal in the $j$th row is on the position
$(j-1)(7\vs+1)+(5\vs+1)$ and
the next possible occurrence of a representation of a falsified literal is at the position
$(j+1-1)(7\vs+1)+(4\vs+2)$ of the following row.
The difference of these two positions is $6\vs+2 \geq 5\vs$.
\qed
\end{proof}

Another intractable restriction of the constraint is when
all values from $S$ are forced to occur on the first $p \leq n$ places,
e.g. in the reduction from the last proof it was the first row.
The following corollary summarizes all found intractable restrictions

\begin{corollary}\label{crl:intract}
Enforcing \DC\ on the \SPACING\ constraint is NP-hard even if
any combination of the following not containing (3) and (4) simultaneously holds:
\begin{enumerate}[(1)]
	\item All the values from $S$ must occur in the first $p \leq n$ places of the sequence.
	\item $A$ and $B$ are constant sequences, i.e. $a_1=\ldots=a_k$ and  $b_1=\ldots=b_k$.
	\item There is no minimal distance condition ($a_i = 1$ for $1 \leq i \leq k-1$).
	\item There is no maximal distance condition ($b_i = n$ for $1 \leq i \leq k-1$).
\end{enumerate}
\end{corollary}

\begin{proof}
The proof of Theorem~\ref{thm:intractB} holds for all combinations of the restrictions
not including (3) and
the proof of Theorem~\ref{thm:intractSF} holds for all combinations of the restrictions
not including (4).
\qed
\end{proof}

\section{Bounded Size of $S$}

We identified two useful restrictions of the $\SPACING$ constraint that 
allow polynomial time $\DC$ filtering algorithms.  The first restriction
bounds the size of $S$, $|S| =  O(1)$. It can be used to model education 
process where the number of learning units is naturally bounded.

\begin{theorem}\label{thm:tractS}
Enforcing $\DC$ on the $\SPACING(S, A, B, X)$ constraint can be done 
in $O(n^{|S|+2}|S|)$ time.
\end{theorem}

\begin{proof}
We can define an automaton for accepting sequences satisfying the
$\SPACING$ constraint. The states of the automaton just need to keep count of
the number of steps since the last occurrence of each value in $S$. There are $O(n^{|S|})$
possible states in this automaton, which is polynomial for $|S| = O(1)$.\qed
\end{proof}

% \TODO[Peter]{more details .?} %TODO NO!!

\section{The One Voice Global \SPACING\ Constraint (\SPACINGONE)}

The second tractable restriction of the $\SPACING$ constraint ensures that
all values from $S$ occur in the first \emph{period} of length $p$ 
and they must repeat in the successive $k-1$ periods on the same \emph{places}.
In other words, the first period of length $p$ is cycled $k$ times.
This restriction is useful in music composition problems~\cite{MusicBook},
where the  composer wants to generate one voice consisting of
a $p$ beat long rhythmical pattern that is played $k$ times.
The pattern consists of $|S|$ onsets (beginnings of notes) that
must be played exactly $k$ times in the whole voice.
This can be encoded using a restriction of \SPACING\ that is defined as follows:

\begin{definition}\label{def:spacing1}
Let $X = \seq{X_1, \ldots, X_n}$ be a sequence of $n$ variables and
let $S \subseteq D$ be a set of domain values.
Let $p$ and $k$ be natural numbers such that $p \leq n$ and $pk \leq n$.
Then $\SPACINGONE(S, p, k, X)$ holds iff
$\SPACING(S, A, B, X)$ with $S \subseteq \set{X_1, \ldots, X_p}$,
$a_i = b_i = p$ for $1 \leq i \leq k-1$ and
$|\setcomp{j}{X_j = d, 1 \leq j \leq n}| = k$ for all $d \in S$ holds.
\end{definition}

\begin{theorem}\label{thm:tract1}
For any constraint $\SPACINGONE(S, p, k, X)$,
there is a bipartite graph $G = (U, V, E)$ such that 
there is a support for the constraint iff there is a maximum matching in $G$.
Enforcing \DC\ on the constraint takes $O(p^2 k + p^{2.5})$ time
down a branch of the search tree.
\end{theorem}

\begin{proof}
First, we observe that values $D \setminus S$ are interchangeable
as we do not distinguish between values $d$ outside $S$, $d \notin S$. Therefore, we perform
channeling of variables $X$ to variables $Y$ and map all values outside of $S$
into a dummy value $\zero$:  
%To satisfy \SPACINGONE, each value from $S$ has to be assigned to
%a variable on the same place in each period of the sequence and
%the other places have to be filled with values not in $S$.
%In order to check this, we channel values not in $S$ to a dummy value
%and make intersections of the places where a value from $S$ must be repeated,
%which we call \emph{folding}.
%To obtain $G$, we first channel $X$ into a new sequence $Y = [Y_i, \ldots, Y_n]$, 
%such that
$X_i \in S \leftrightarrow Y_i = X_i$ and $X_i \not\in S \leftrightarrow Y_i = \zero$ for $1 \leq i \leq n$,
where $\zero \notin S$ is a fresh value.

Second, we exploit the special structure of $\SPACINGONE$. Namely, variables in
positions $i, p + i,\ldots, (k-1)p + i$, $1 \leq i \leq p$, must take the same value.
Hence, to check whether  value $d$ can be assigned to one of the 
variables $Y_i, Y_{p + i},\ldots, Y_{(k-1)p + i}$, we need to check 
whether $d \in \bigcap_{j=0}^{k-1} D(Y_{j p + i})$. We use a \emph{folding}
procedure to identify possible positions for each value.
We \emph{fold} the domains of $Y$ into $P_i = \bigcap_{j=0}^{k-1} D(Y_{j p + i})$ for $1 \leq i \leq p$.

Finally, we need to match values $S \cup \set{\zero}$  with the positions in one period.
To avoid using  generalized matching, we 
%As long as each place have to be matched with some value, we need 
introduce $p - |S|$ copies of the dummy value, otherwise we would have
to match $\zero$ with $p - |S|$ nodes.

Next we describe  construction of the graph $G$.
The sets of nodes are $U = S \cup \setcomp{\zero_j}{1 \leq j \leq p - |S|}$, $V = \set{1, \ldots, p}$.
The set of edges is $E = \setcomp{(d, i)}{d \in P_i \cap S, 1 \leq i \leq p} \cup
\setcomp{(\zero_j, i)}{1 \leq j \leq p - |S|, \zero \in P_i, 1 \leq i \leq p}$.
Now we show that there exists a support for $\SPACINGONE$ iff there exists
a maximum matching in $G$.

Having a support for the constraint, we can obtain a subgraph of $G$
with the same sets of nodes and the set of edges $M$ that consists of two parts:
\begin{inparaenum}[(1)]
	\item the edges between the places of the variables in the first period
	and the values from $S$ that are assigned to them in the support
	$M_1 = \setcomp{(d,i)}{d = X_i, d \in S, 1 \leq i \leq p}$,
	\item if we order the rest of the places in the first period
	(that contain values not in $S$) into a new sequence
	$N = \left[ i | X_i \notin S, 1 \leq i \leq p \right]$ (note that $|N| = p - |S|$),
	the edges between these places and the dummy values with the same index
	$M_2 = \setcomp{(\zero_j,N_j)}{1 \leq j \leq p - |S|}$.
\end{inparaenum}
So $M = M_1 \cup M_2$.
% \TODO{M is a subgraph} %TODO
Now we have indeed a subgraph of $G$, because $M$ is a subset of $E$, because,
while we have a support, values from $S$ have to repeat on the same places in each period, so
$X_i \in P_i$ if $X_i \in S$ and
$\zero \in P_i$ if $X_i \notin S$ for $1 \leq i \leq p$, thus
$M_1 \subseteq \setcomp{(d, i)}{d \in P_i \cap S, 1 \leq i \leq p}$ and
$M_2 \subseteq \setcomp{(\zero_j, i)}{1 \leq j \leq p - |S|, \zero \in P_i, 1 \leq i \leq p}$.
% \TODO{M is a matching} %TODO
$M$ is also a matching in $G$, because exactly one value is assigned to each variable and
each dummy value is connected to exactly one place.
% \TODO{M is a maximum matching} %TODO
$M$ is a maximum matching, because it contains all the nodes, because
each value from $S$ have to be assigned to some variable in the first period and
the number of the dummy values is exactly the number of variables in the first period
with values not in $S$ assigned to them.
Hence, we obtained a maximum matching in $G$ from a support for the constraint.

Now having a maximum matching $M$ in $G$, we can obtain a support for the constraint by
assigning the values from $S$ to the places in each period to which they are matched by $M$
and any value not in $S$ to the places in each period to which a dummy value is matched by $M$.
This is a valid assignment, because $M$ matches something with each place, because it is maximal.
This assignment is always possible, because $G$ connects each $d \in S$
only to places of variables with domains containing $d$ in each period and
$G$ connects dummy values only to places of variables with domains containing
a value not in $S$ in each period. The assignment is indeed a support, because
each value from $S$ is matched with exactly one place of a period and repeats
altogether $k$ times on the same place in each period.

Our $\DC$ propagator is based on $\DC$ propagator for $\AllDifferent$ by Regin\cite{regin1}.
First, we determine the set of edges that do not belong to any maximum matching
the same way as propagator for $\AllDifferent$ enforces $\DC$. This takes
$O(p^{2.5})$ down a branch of the search tree. If an edge $(u,i)$ does not belong
to any matching then the value $u$ can be removed from domains  
of variables $Y_i, Y_{p + i},\ldots, Y_{(k-1)p + i}$ which takes $O(k)$ time.
%Matching a dummy value $\zero_j$ with a place $i$ means removing $S$ from $D(X_i)$.
%The matching part is the same as the filtering algorithm for \ by , %TODO
%hence its complexity is $O(p^{2.5})$. 
%Additionally, we have to repeat the removals $k$ times
%and
There can be at most $O(p^2)$ removals down a branch, so
overall complexity down a branch of the search tree is $O(p^2 k + p^{2.5})$.
\qed
\end{proof}

For example, consider the \SPACINGONE\ constraint with $D = \set{a,b,c,o}$,
$S = \set{a,b,c}$ representing the onsets, $p = 5$ length of the pattern and
$k = 3$ number of repetitions on the sequence of 15 variables. 
The variable domains are as shown below:
\begin{center}
\newcommand{\compd}[1]{\scriptsize\begin{tabular}{>{$}c<{$}} #1 \end{tabular}}
\begin{tabular}{r|*{15}{|>{\centering\arraybackslash}p{3.5mm}}}
$i = $ & 1 & 2 & 3 & 4 & 5 & 6 & 7 & 8 & 9 & 10 & 11 & 12 & 13 & 14 & 15 \\
\hline
\hline
$D(X_i) =$ &
\compd{ a \\ b \\ \ \\ o \\ } &
\compd{ a \\ b \\ c \\ o \\ } &
\compd{ a \\ b \\ c \\ o \\ } &
\compd{ a \\ b \\ \ \\ \ \\ } &
\compd{ \ \\ b \\ c \\ o \\ } &
\compd{ a \\ b \\ c \\ o \\ } &
\compd{ a \\ b \\ c \\ \ \\ } &
\compd{ \ \\ \ \\ c \\ \ \\ } &
\compd{ a \\ b \\ c \\ o \\ } &
\compd{ \ \\ b \\ c \\ o \\ } &
\compd{ a \\ b \\ c \\ o \\ } &
\compd{ a \\ \ \\ c \\ o \\ } &
\compd{ a \\ b \\ c \\ o \\ } &
\compd{ a \\ b \\ c \\ o \\ } &
\compd{ \ \\ b \\ c \\ o \\ } \\
\end{tabular}
\end{center}
During propagation, channeling will simply replaces $o$ with \zero.
Then the folded domains will look like this:
\begin{center}
\newcommand{\compd}[1]{\scriptsize\begin{tabular}{>{$}c<{$}} #1 \end{tabular}}
\begin{tabular}{r|*{5}{|>{\centering\arraybackslash}p{3.5mm}}}
$i = $ & 1 & 2 & 3 & 4 & 5 \\
\hline
\hline
$P_i =$ &
\compd{ a \\ b \\ \ \\ \zero \\ } &
\compd{ a \\ \ \\ c \\ \ \\ } &
\compd{ \ \\ \ \\ c \\ \ \\ } &
\compd{ a \\ b \\ \ \\ \ \\ } &
\compd{ \ \\ b \\ c \\ \zero \\ } \\
\end{tabular}
\end{center}
And finally, the bipartite graph is given in Figure~\ref{f:f1}.
There are two maximum matchings in this graph:
$\set{(c,3), (a,2), (b,4), (\zero_1, 1), (\zero_2, 5)}$ and $\set{(c,3), (a,2), (b,4), (\zero_2, 1), (\zero_1, 5)}$.
This means that the only support is $\seq{o,a,c,b,o,o,a,c,b,o,o,a,c,b,o}$.

\begin{figure}[htb] \centering
{\scriptsize
	\caption{\label{f:f1}A bipartite graph for $\SPACINGONE(\set{a,b,c}, 5, 3, \seq{X_1, \ldots, X_{15}})$ from the example}
}
\begin{tikzpicture} [
		every circle node/.style={draw,inner sep=0,minimum size=5mm},
		anchor=base,
 		transform shape,
 		scale=0.8,%rotate=10,
	]

	\path (0,2.5) -- (0,0)
		node[circle] (1) [pos=0] {$a$}
		node[circle] (2) [pos=0.25] {$b$}
		node[circle] (3) [pos=0.5] {$c$}
		node[circle] (01) [pos=0.75] {$\zero_1$}
		node[circle] (02) [pos=1] {$\zero_2$};
	\path (3,2.5) -- (3,0)
		node[circle] (x1) [pos=0] {1}
		node[circle] (x2) [pos=0.25] {2}
		node[circle] (x3) [pos=0.5] {3}
		node[circle] (x4) [pos=0.75] {4}
		node[circle] (x5) [pos=1] {5};

	\foreach \n in {1,2,4} {\draw (1) -- (x\n);};
	\foreach \n in {1,4,5} {\draw (2) -- (x\n);};
	\foreach \n in {2,3,5} {\draw (3) -- (x\n);};
	\foreach \n in {1,5} {\draw (01) -- (x\n);};
	\foreach \n in {1,5} {\draw (02) -- (x\n);};

\end{tikzpicture}
\end{figure}

\section{The $h$-Voice Global \SPACING\ Constraint (\SPACINGH)}

A composer would usually want to compose more voices playing at the same time
with no overlapping onsets. The $h$-Voice Global \SPACING\ Constraint,
which is simply a conjunction of more \SPACINGONE\ constraints on the same
sequence of variables, can be used for this purpose.

\begin{definition}[\SPACINGH]\label{def:spacingh}
Let $X = \seq{X_1, \ldots, X_n}$ be a sequence of $n$ variables and
let $S = \seq{S_1, \ldots, S_h}$ be a sequence of pairwise disjoint sets of domain values
($\bigcup_{l = 1}^{h} S_l \subseteq D$ and
$S_{l_1} \cap S_{l_2} = \emptyset$ for all $1 \leq l_1 < l_2 \leq h$).
Let $\seq{p_1, \ldots, p_h}$ be a sequence of natural numbers and
let $k$ be a natural number such that $p_{l}k \leq n$ for all $1 \leq l \leq h$.
Then $\SPACINGH(S, \seq{p_1, \ldots, p_h}, k, X)$ holds iff
a conjunction of constraints $\SPACINGONE(S_l, p_l, k, X)$ for all $1 \leq l \leq h$ holds.
\end{definition}

\begin{theorem}\label{thm:intractSH}
Enforcing \DC\ on the Global \SPACINGH\ Constraint is NP-hard even for $h=2$.
\end{theorem}

\begin{proof}
We prove this by reduction of \SAT\ to the problem of finding a support for \SPACINGH.
The main idea of the proof is similar to the other hardness proofs:
The alternative choice of a literal satisfying a clause is modeled by
choice of a value for a variable and mutual exclusion of complementary literals is
enforced by properties of \SPACING.
For the \SPACINGH\ constraint it is mutual exclusion of values from the same variable domain
(assignment of one excludes the others).
The system how choosing two complementary literals to be true leads to
the mutual exclusion of values is quite complicated, because
we have little freedom (just two voices).
The idea is that the literal chosen to satisfy a clause is copied to
the part representing the model. This is done by the impossibility of
having more than one occurrence of the same value in one period of a voice, so
from domains of cardinality 2 containing this value the other one has to be chosen.
These other values (that represent the model) belong, however, to the other voice,
which has period of different length, so repetitions of these values will be aligned
with different places of the other periods of the first voice.
On the one hand, this is used for copying the clause satisfiers to the same model representation,
on the other hand, it is used to align values for complementary literals with each other so
that their assignment mutually exclude each other.

Now we describe the reduction in detail.
Please, recall that if $L$ is a set of literals, then
$L' = \setcomp{l'}{l \in L}$ and $L^i = \setcomp{l^i}{l \in L}$.
Let $\phi$ be an arbitrary \CNF\ with $\vs$ propositional variables and $\cs$ clauses.
There is a model of $\phi$ iff there is a support for
the constraint $\SPACINGH(\seq{S_1, S_2}, \seq{p_1, p_2}, k, X)$ where
\begin{itemize}
	\item $S_1 = \lit(\phi)' \cup \bigcup_{j = 1}^\cs \left( \lit(\phi)^j \right)$
	\item $S_2 = \lit(\phi)$
	\item $p_1 = \cs + 6\cs\vs$
	\item $p_2 = \cs + 6\cs\vs + 2\vs$
	\item $k = \cs$
	\item $X$ is a sequence of $(\cs + 6\cs\vs + 2\vs)k$ variables with the domains as described below.
\end{itemize}
We will describe the domains voice by voice,
so ultimately the domains are minimal sets satisfying the following conditions.
The period of the first voice is $p_1 = \cs + 6\cs\vs$,
let us split the sequence into rows of this length
and order them into a table (one row is one period of the first voice).
Let $X^1_{j,i}$ denote $X_{(j-1)p_1 + i}$.
Then the first $\cs$ columns represent $\phi$ in the following fashion:
The cells on the main diagonal represent the clauses and
the other cells are filled with representation of all the literals,
\begin{itemize}
	\item the domains of variables $X^1_{j,i}$ for $1 \leq j,i \leq \cs$ contain
	$\setcomp{p^i}{p \in C_j}$ if $i = j$,
	otherwise $\lit(\phi)^i$.
\end{itemize}
The following $4\cs\vs$ columns contain representations of all literals for each clause two times,
\begin{itemize}
	\item the domains of variables $X^1_{j,(\cs) + (i-1)2\vs + \vs y + x}$ for $1 \leq j,i \leq \cs$, $0 \leq y \leq 1$ and $1 \leq x \leq \vs$
	contain $p^i_x$ if $y=0$ and $\neg p^i_x$ if $y=1$,
	\item the domains of variables $X^1_{j,(\cs + 2\cs\vs) + (i-1)2\vs + \vs y + x}$ for $1 \leq j,i \leq \cs$, $0 \leq y \leq 1$ and $1 \leq x \leq \vs$
	contain $p^i_x$ if $y=0$ and $\neg p^i_x$ if $y=1$.
\end{itemize}
The last $2\cs\vs$ columns contain representations of all literals
in the usual order in their first $2\vs$ columns,
in their last $2\vs$ in the order where positive and negative literals are swapped,
and \zero\ in the rest of their columns,
\begin{itemize}
	\item the domains of variables $X^1_{j,(\cs + 4\cs\vs) + x}$ for
	$1 \leq j \leq \cs$ and $1 \leq x \leq \vs$ contain $p'_x$,
	\item the domains of variables $X^1_{j,(\cs + 4\cs\vs) + \vs + x}$ for
	$1 \leq j \leq \cs$ and $1 \leq x \leq \vs$ contain $\neg p'_x$,
	\item the domains of variables $X^1_{j,(\cs + 4\cs\vs) + 2\vs + x}$ for
	$1 \leq j \leq \cs$ and $1 \leq x \leq 2\cs\vs - 4\vs$ contain $\zero$,
	\item the domains of variables $X^1_{j,(\cs + 4\cs\vs) + (2\cs\vs-2\vs) + x}$ for
	$1 \leq j \leq \cs$ and $1 \leq x \leq \vs$ contain $\neg p'_x$,
	\item the domains of variables $X^1_{j,(\cs + 4\cs\vs) + (2\cs\vs-\vs) + x}$ for
	$1 \leq j \leq \cs$ and $1 \leq x \leq \vs$ contain $p'_x$.
\end{itemize}
(This is not well defined for $\cs=1$, but $\phi$ is trivially satisfiable in this case.)
The arrangement of the values of the second voice is much more simple.
The period of the second voice is $p_2 = \cs + 6\cs\vs + 2\vs$, so
let $X^2_{j,i}$ denote $X_{(j-1)p_2 + i}$.
Then places from $\cs+1$ to $\cs+2\vs$ and from $\cs+4\cs\vs+1$ to $\cs+4\cs\vs+2\vs$ of each period
contain representations of all literals in the usual order,
\begin{itemize}
	\item the domains of variables $X^2_{j,(\cs) + x}$ for
	$1 \leq j \leq \cs$ and $1 \leq x \leq \vs$ contain $p_x$,
	\item the domains of variables $X^2_{j,(\cs + \vs) + x}$ for
	$1 \leq j \leq \cs$ and $1 \leq x \leq \vs$ contain $\neg p_x$,
	\item the domains of variables $X^2_{j,(\cs + 4\cs\vs) + x}$ for
	$1 \leq j \leq \cs$ and $1 \leq x \leq \vs$ contain $p_x$,
	\item the domains of variables $X^2_{j,(\cs + 4\cs\vs + \vs) + x}$ for
	$1 \leq j \leq \cs$ and $1 \leq x \leq \vs$ contain $\neg p_x$.
\end{itemize}
Finally, the domains of variables that
are not on the first $\cs$ places of any period of the first voice
and do not contain a value from $S_2$ contain \zero:
\begin{itemize}
	\item the domains of variables $X_i$ for
	$i \notin \setcomp{(j-1)p_1 + x}{1 \leq j,x \leq \cs} \cup
	\setcomp{(j-1)p_2 + \cs + x}{1 \leq j \leq \cs, 1 \leq x \leq 2\vs} \cup
	\setcomp{(j-1)p_2 + \cs + 4\cs\vs + x}{1 \leq j \leq \cs, 1 \leq x \leq 2\vs}$
	contain \zero.
\end{itemize}

The constraint for the CNF from our running example would
be imposed on a sequence of variables with domains displayed in the Table~\ref{t:exampleh}.
The last $2\vs k$ domains of $\set{\zero}$ are not displayed in the table.

\begin{sidewaystable}
\caption{
	\label{t:exampleh}Example for reduction from \SAT\ to \SPACINGH.
}
\begin{center}

\newcommand{\compd}[1]{\scriptsize\begin{tabular}{>{$}c<{$}} #1 \end{tabular}}
\begin{tabular}{r|*{4}{|>{\centering\arraybackslash}c}|*{24}{|>{\centering\arraybackslash}c}}
 & 1 & 2 & 3 & 4 & 5 & 6 & 7 & 8 & 9 & 10 & 11 & 12 & 13 & 14 & 15 & 16 & 17 & 18 & 19 & 20 & 21 & 22 & 23 & 24 & 25 & 26 & 27 & 28 \\
\hline
\hline
1 & \compd{\neg p^{1} \\ q^{1} \\ r^{1}} & \compd{p^{2} \\ q^{2} \\ r^{2} \\ \neg p^{2} \\ \neg q^{2} \\ \neg r^{2}} & \compd{p^{3} \\ q^{3} \\ r^{3} \\ \neg p^{3} \\ \neg q^{3} \\ \neg r^{3}} & \compd{p^{4} \\ q^{4} \\ r^{4} \\ \neg p^{4} \\ \neg q^{4} \\ \neg r^{4}} & \compd{p^{1} \\ p} & \compd{q^{1} \\ q} & \compd{r^{1} \\ r} & \compd{\neg p^{1} \\ \neg p} & \compd{\neg q^{1} \\ \neg q} & \compd{\neg r^{1} \\ \neg r} & \compd{p^{2} \\ \zero} & \compd{q^{2} \\ \zero} & \compd{r^{2} \\ \zero} & \compd{\neg p^{2} \\ \zero} & \compd{\neg q^{2} \\ \zero} & \compd{\neg r^{2} \\ \zero} & \compd{p^{3} \\ \zero} & \compd{q^{3} \\ \zero} & \compd{r^{3} \\ \zero} & \compd{\neg p^{3} \\ \zero} & \compd{\neg q^{3} \\ \zero} & \compd{\neg r^{3} \\ \zero} & \compd{p^{4} \\ \zero} & \compd{q^{4} \\ \zero} & \compd{r^{4} \\ \zero} & \compd{\neg p^{4} \\ \zero} & \compd{\neg q^{4} \\ \zero} & \compd{\neg r^{4} \\ \zero} \\
\hline
2 & \compd{p^{1} \\ q^{1} \\ r^{1} \\ \neg p^{1} \\ \neg q^{1} \\ \neg r^{1}} & \compd{\neg q^{2} \\ r^{2}} & \compd{p^{3} \\ q^{3} \\ r^{3} \\ \neg p^{3} \\ \neg q^{3} \\ \neg r^{3}} & \compd{p^{4} \\ q^{4} \\ r^{4} \\ \neg p^{4} \\ \neg q^{4} \\ \neg r^{4}} & \compd{p^{1} \\ \zero} & \compd{q^{1} \\ \zero} & \compd{r^{1} \\ \zero} & \compd{\neg p^{1} \\ \zero} & \compd{\neg q^{1} \\ \zero} & \compd{\neg r^{1} \\ \zero} & \compd{p^{2} \\ p} & \compd{q^{2} \\ q} & \compd{r^{2} \\ r} & \compd{\neg p^{2} \\ \neg p} & \compd{\neg q^{2} \\ \neg q} & \compd{\neg r^{2} \\ \neg r} & \compd{p^{3} \\ \zero} & \compd{q^{3} \\ \zero} & \compd{r^{3} \\ \zero} & \compd{\neg p^{3} \\ \zero} & \compd{\neg q^{3} \\ \zero} & \compd{\neg r^{3} \\ \zero} & \compd{p^{4} \\ \zero} & \compd{q^{4} \\ \zero} & \compd{r^{4} \\ \zero} & \compd{\neg p^{4} \\ \zero} & \compd{\neg q^{4} \\ \zero} & \compd{\neg r^{4} \\ \zero} \\
\hline
3 & \compd{p^{1} \\ q^{1} \\ r^{1} \\ \neg p^{1} \\ \neg q^{1} \\ \neg r^{1}} & \compd{p^{2} \\ q^{2} \\ r^{2} \\ \neg p^{2} \\ \neg q^{2} \\ \neg r^{2}} & \compd{\neg p^{3} \\ \neg q^{3}} & \compd{p^{4} \\ q^{4} \\ r^{4} \\ \neg p^{4} \\ \neg q^{4} \\ \neg r^{4}} & \compd{p^{1} \\ \zero} & \compd{q^{1} \\ \zero} & \compd{r^{1} \\ \zero} & \compd{\neg p^{1} \\ \zero} & \compd{\neg q^{1} \\ \zero} & \compd{\neg r^{1} \\ \zero} & \compd{p^{2} \\ \zero} & \compd{q^{2} \\ \zero} & \compd{r^{2} \\ \zero} & \compd{\neg p^{2} \\ \zero} & \compd{\neg q^{2} \\ \zero} & \compd{\neg r^{2} \\ \zero} & \compd{p^{3} \\ p} & \compd{q^{3} \\ q} & \compd{r^{3} \\ r} & \compd{\neg p^{3} \\ \neg p} & \compd{\neg q^{3} \\ \neg q} & \compd{\neg r^{3} \\ \neg r} & \compd{p^{4} \\ \zero} & \compd{q^{4} \\ \zero} & \compd{r^{4} \\ \zero} & \compd{\neg p^{4} \\ \zero} & \compd{\neg q^{4} \\ \zero} & \compd{\neg r^{4} \\ \zero} \\
\hline
4 & \compd{p^{1} \\ q^{1} \\ r^{1} \\ \neg p^{1} \\ \neg q^{1} \\ \neg r^{1}} & \compd{p^{2} \\ q^{2} \\ r^{2} \\ \neg p^{2} \\ \neg q^{2} \\ \neg r^{2}} & \compd{p^{3} \\ q^{3} \\ r^{3} \\ \neg p^{3} \\ \neg q^{3} \\ \neg r^{3}} & \compd{p^{4} \\ q^{4}} & \compd{p^{1} \\ \zero} & \compd{q^{1} \\ \zero} & \compd{r^{1} \\ \zero} & \compd{\neg p^{1} \\ \zero} & \compd{\neg q^{1} \\ \zero} & \compd{\neg r^{1} \\ \zero} & \compd{p^{2} \\ \zero} & \compd{q^{2} \\ \zero} & \compd{r^{2} \\ \zero} & \compd{\neg p^{2} \\ \zero} & \compd{\neg q^{2} \\ \zero} & \compd{\neg r^{2} \\ \zero} & \compd{p^{3} \\ \zero} & \compd{q^{3} \\ \zero} & \compd{r^{3} \\ \zero} & \compd{\neg p^{3} \\ \zero} & \compd{\neg q^{3} \\ \zero} & \compd{\neg r^{3} \\ \zero} & \compd{p^{4} \\ p} & \compd{q^{4} \\ q} & \compd{r^{4} \\ r} & \compd{\neg p^{4} \\ \neg p} & \compd{\neg q^{4} \\ \neg q} & \compd{\neg r^{4} \\ \neg r} \\
\end{tabular}

\bigskip

\begin{tabular}{r|*{9}{|>{\centering\arraybackslash}c}|*{24}{|>{\centering\arraybackslash}c}}
 & 29 & 30 & 31 & 32 &  & 49 & 50 & 51 & 52 & 53 & 54 & 55 & 56 & 57 & 58 & 59 & 60 & 61 & 62 & 63 & 64 & 65 & 66 & 67 & 68 & 69 & 70 & 71 & 72 & 73 & 74 & 75 & 76 \\
\hline
\hline
1 & \compd{p^{1} \\ \zero} & \compd{q^{1} \\ \zero} & \compd{r^{1} \\ \zero} & \compd{\neg p^{1} \\ \zero} & \compd{\ldots} & \compd{r^{4} \\ \zero} & \compd{\neg p^{4} \\ \zero} & \compd{\neg q^{4} \\ \zero} & \compd{\neg r^{4} \\ \zero} & \compd{p' \\ p} & \compd{q' \\ q} & \compd{r' \\ r} & \compd{\neg p' \\ \neg p} & \compd{\neg q' \\ \neg q} & \compd{\neg r' \\ \neg r} & \compd{\zero} & \compd{\zero} & \compd{\zero} & \compd{\zero} & \compd{\zero} & \compd{\zero} & \compd{\zero} & \compd{\zero} & \compd{\zero} & \compd{\zero} & \compd{\zero} & \compd{\zero} & \compd{\neg p' \\ \zero} & \compd{\neg q' \\ \zero} & \compd{\neg r' \\ \zero} & \compd{p' \\ \zero} & \compd{q' \\ \zero} & \compd{r' \\ \zero} \\
\hline
2 & \compd{p^{1} \\ \zero} & \compd{q^{1} \\ \zero} & \compd{r^{1} \\ \zero} & \compd{\neg p^{1} \\ \zero} & \compd{\ldots} & \compd{r^{4} \\ \zero} & \compd{\neg p^{4} \\ \zero} & \compd{\neg q^{4} \\ \zero} & \compd{\neg r^{4} \\ \zero} & \compd{p' \\ \zero} & \compd{q' \\ \zero} & \compd{r' \\ \zero} & \compd{\neg p' \\ \zero} & \compd{\neg q' \\ \zero} & \compd{\neg r' \\ \zero} & \compd{\zero \\ p} & \compd{\zero \\ q} & \compd{\zero \\ r} & \compd{\zero \\ \neg p} & \compd{\zero \\ \neg q} & \compd{\zero \\ \neg r} & \compd{\zero} & \compd{\zero} & \compd{\zero} & \compd{\zero} & \compd{\zero} & \compd{\zero} & \compd{\neg p' \\ \zero} & \compd{\neg q' \\ \zero} & \compd{\neg r' \\ \zero} & \compd{p' \\ \zero} & \compd{q' \\ \zero} & \compd{r' \\ \zero} \\
\hline
3 & \compd{p^{1} \\ \zero} & \compd{q^{1} \\ \zero} & \compd{r^{1} \\ \zero} & \compd{\neg p^{1} \\ \zero} & \compd{\ldots} & \compd{r^{4} \\ \zero} & \compd{\neg p^{4} \\ \zero} & \compd{\neg q^{4} \\ \zero} & \compd{\neg r^{4} \\ \zero} & \compd{p' \\ \zero} & \compd{q' \\ \zero} & \compd{r' \\ \zero} & \compd{\neg p' \\ \zero} & \compd{\neg q' \\ \zero} & \compd{\neg r' \\ \zero} & \compd{\zero} & \compd{\zero} & \compd{\zero} & \compd{\zero} & \compd{\zero} & \compd{\zero} & \compd{\zero \\ p} & \compd{\zero \\ q} & \compd{\zero \\ r} & \compd{\zero \\ \neg p} & \compd{\zero \\ \neg q} & \compd{\zero \\ \neg r} & \compd{\neg p' \\ \zero} & \compd{\neg q' \\ \zero} & \compd{\neg r' \\ \zero} & \compd{p' \\ \zero} & \compd{q' \\ \zero} & \compd{r' \\ \zero} \\
\hline
4 & \compd{p^{1} \\ \zero} & \compd{q^{1} \\ \zero} & \compd{r^{1} \\ \zero} & \compd{\neg p^{1} \\ \zero} & \compd{\ldots} & \compd{r^{4} \\ \zero} & \compd{\neg p^{4} \\ \zero} & \compd{\neg q^{4} \\ \zero} & \compd{\neg r^{4} \\ \zero} & \compd{p' \\ \zero} & \compd{q' \\ \zero} & \compd{r' \\ \zero} & \compd{\neg p' \\ \zero} & \compd{\neg q' \\ \zero} & \compd{\neg r' \\ \zero} & \compd{\zero} & \compd{\zero} & \compd{\zero} & \compd{\zero} & \compd{\zero} & \compd{\zero} & \compd{\zero} & \compd{\zero} & \compd{\zero} & \compd{\zero} & \compd{\zero} & \compd{\zero} & \compd{\neg p' \\ p} & \compd{\neg q' \\ q} & \compd{\neg r' \\ r} & \compd{p' \\ \neg p} & \compd{q' \\ \neg q} & \compd{r' \\ \neg r} \\
\end{tabular}

\end{center}
\end{sidewaystable}

Let us organize the sequence $X$ into a table with $\cs + 6\cs\vs$ columns again
(using $X^1_{j,i}$ notation) and ignore the last $2\cs\vs$ variables with domains $\set{\zero}$.
Each row belongs to one clause.
Also there is one set of literal representing values for each clause ($\lit(\phi)^j$ for $C_j$).
Clauses form domains on the main diagonal using their sets of values.
The rest of the domains in the first $\cs$ columns are constructed to
enable repetitions of values selected for the clause variables.
The next $2\cs\vs$ columns are called \emph{positive} part,
because the values from $S_2$ it contains will represent a model of $\phi$.
Due to the difference between periods, in each row $j$
these values from $S_2$ are aligned with the values of the clause of the row (values from $\lit(\phi)^j$).
Thanks to this, when a value $d^j$ is assigned to the variable of clause $C_j$,
$d^j$ cannot occur in the rest of the row, thus also not in the positive part,
so the appropriate value $d$ has to be selected in this part.
This is why, having a solution of the constraint,
at least one literal from each clause will be assigned as a value from $S_2$
to some variable in the positive part.

The last $2\cs\vs$ columns are \emph{consistency} part, because
they ensure the consistency of the model represented by a solution.
When a value $d \in S_2$ is assigned to a variable in the positive part,
it has to be repeated in the positive part in each row.
Also it cannot be repeated in the rest of any row.
Hence a corresponding primed value $d' \in S_1$ must be assigned
to one of the first $2\vs$ places of the consistency part of the first row (and also each other row),
thus $d'$ cannot repeat in the last $2\vs$ places of the consistency part in any row,
mainly not in the last one.
Please note that the positive and negative primed values are switched in the last $2\vs$ places,
so that, in the last row, the negative primed values are aligned with the positive values from $S_2$
and the positive ones are aligned with the negative ones.
This whole construction of the domains causes that if two complementary literal values
$d, \overline{d} \in S_2$ are assigned to variables in the positive part,
they cannot be assigned to the first $2\vs$ variables in the first row of the consistency part,
so both $d', \overline{d'} \in S_1$ have to be assigned in these first $2\vs$ variables and
cannot be assigned to the last $2\vs$ variables in the last row of the consistency part.
However, the same assignment of values from $S_2$ must be repeated in each row of this part,
so two variables in the last $2\vs$ places of the last row would be left with empty domains.
Analogous reasoning rules out the case when
none of $d, \overline{d} \in S_2$ is assigned in the positive part.
Also it is easy to see that when
exactly one of $d, \overline{d} \in S_2$ is assigned in the positive part,
no contradiction is reached.
This shows that, having a solution of the constraint,
the set of literals assigned as a values from $S_2$ in the positive part is consistent.
So, together with the result above, having a solution of the constraint,
the set of literals assigned as a values from $S_2$ in the positive part is a model of $\phi$.

We can obtain a solution of the constraint from a model of $\phi$ in the following way:
We assign the literals from the model as a values of the second voice to the variables in the positive part.
Then we choose a value that represents some literal satisfied by the model for the clause variables.
And assign the rest of the variables so that the constraint holds.
It is always possible to do this.
It is always possible to assign the literals from the model to the positive part, because
there is one variable for each literal in each row of the part and
these variables repeat with the period of the second voice (where the values are from).
It is always possible to complete the assignment to the positive part, because
each value has an alternative in each domain and each value can be repeated with it's period.
As long as the model is consistent, no contradiction can be reached in the consistency part
because of the following:
For any pair of complementary literals $d, \overline{d} \in S_2$,
one of them is assigned in the positive part and the other one is not.
The selection of the values from $S_2$ in the consistency part is
reversed with comparison to the positive part, so the mutual exclusion still holds.
The selection of the primed values in the first $2\vs$ columns of the consistency part will
represent the model again and the selection of the primed values in the last $2\vs$ columns
will be reversed again.
Due to the fact that the positive values are swapped for the negative ones in the last $2\vs$ columns,
assignment of the primed values complements the assignment of values from $S_2$
in the first as well in the last $2\vs$ columns and no domains are emptied.
It is always possible to select some value for the clause variables because of the following:
In the positive part of each row $j$,
the selection of the values indexed with $j$ will be complementary
to the selection of the values from the model.
This leaves the possibility for the indexed values representing literals satisfied by the model
to be selected elsewhere in the row, for instance in the clause variable.
Each clause variable always contains at lease one value representing literal satisfied by the model.
Assignment to the clause variables can be repeated in each row and
the primed values that could not be selected in the positive part nor in the clause values
can still be assigned to some of the variables in the columns
between $\cs+2\cs\vs+1$ and $\cs+4\cs\vs$ called \emph{negative} part.
\qed
\end{proof}

\section{Incomplete filtering of \SPACINGH}

Since already a conjunction of two \SPACINGONE\ constraints is NP-hard,
we decided to devise an incomplete additional rule that facilitates propagation between
two \SPACINGONE\ constraints on the same sequence.
If periods of these two constraints are of different length,
assigning a value from the first one to one variable forbids
assigning any value from the second one to multiple variables.
For example, take $\SPACINGONE(\set{a,b}, 5, 4, X)$ and $\SPACINGONE(\set{c,d}, 7, 3, X)$
on the same sequence $X$ of 21 variables.
We start with full domains for all variables.
After assign $a$ to the first variable and filtering the domains are:
\begin{center}
\newcommand{\compd}[1]{\scriptsize\begin{tabular}{>{$}c<{$}} #1 \end{tabular}}
\begin{tabular}{r|*{21}{|>{\centering\arraybackslash}p{3.7mm}}}
$i = $ & 1 & 2 & 3 & 4 & 5 & 6 & 7 & 8 & 9 & 10 & 11 & 12 & 13 & 14 & 15 & 16 & 17 & 18 & 19 & 20 & 21 \\
\hline
\hline
$D(X_i) =$ &
\compd{ a \\ \ \\ \ \\ \ \\ \ } &
\compd{ \ \\ b \\ \ \\ \ \\ \zero } &
\compd{ \ \\ b \\ c \\ d \\ \zero } &
\compd{ \ \\ b \\ \ \\ \ \\ \zero } &
\compd{ \ \\ b \\ c \\ d \\ \zero } &
\compd{ a \\ \ \\ \ \\ \ \\ \ } &
\compd{ \ \\ b \\ c \\ d \\ \zero } &
\compd{ \ \\ b \\ \ \\ \ \\ \zero } &
\compd{ \ \\ b \\ \ \\ \ \\ \zero } &
\compd{ \ \\ b \\ c \\ d \\ \zero } &
\compd{ a \\ \ \\ \ \\ \ \\ \ } &
\compd{ \ \\ b \\ c \\ d \\ \zero } &
\compd{ \ \\ b \\ \ \\ \ \\ \zero } &
\compd{ \ \\ b \\ c \\ d \\ \zero } &
\compd{ \ \\ b \\ \ \\ \ \\ \zero } &
\compd{ a \\ \ \\ \ \\ \ \\ \ } &
\compd{ \ \\ b \\ c \\ d \\ \zero } &
\compd{ \ \\ b \\ \ \\ \ \\ \zero } &
\compd{ \ \\ b \\ c \\ d \\ \zero } &
\compd{ \ \\ b \\ \ \\ \ \\ \zero } &
\compd{ \ \\ \ \\ c \\ d \\ \zero } \\
\end{tabular}
\end{center}
If we assign $b$ to $X_4$, the assignment must be repeated every 5 variables.
These repetitions occur on different places of the second period,
so sufficient number of repetitions of $c$ and $d$ will not be possible.
Subsequent filtering would remove circled values
\begin{center}
\newcommand{\compd}[1]{\scriptsize\begin{tabular}{>{$}c<{$}} #1 \end{tabular}}
\begin{tabular}{r|*{21}{|>{\centering\arraybackslash}p{3.7mm}}}
$i = $ & 1 & 2 & 3 & 4 & 5 & 6 & 7 & 8 & 9 & 10 & 11 & 12 & 13 & 14 & 15 & 16 & 17 & 18 & 19 & 20 & 21 \\
\hline
\hline
$D(X_i) =$ &
\compd{ a \\ \ \\ \ \\ \ \\ \ } &
\compd{ \ \\ \textcircled b \\ \ \\ \ \\ \zero } &
\compd{ \ \\ \textcircled b \\ c \\ d \\ \zero } &
\compd{ \ \\ b \\ \ \\ \ \\ \textcircled \zero } &
\compd{ \ \\ \textcircled b \\ \textcircled c \\ \textcircled d \\ \zero } &
\compd{ a \\ \ \\ \ \\ \ \\ \ } &
\compd{ \ \\ \textcircled b \\ \textcircled c \\ \textcircled d \\ \zero } &
\compd{ \ \\ \textcircled b \\ \ \\ \ \\ \zero } &
\compd{ \ \\ b \\ \ \\ \ \\ \textcircled \zero } &
\compd{ \ \\ \textcircled b \\ c \\ d \\ \zero } &
\compd{ a \\ \ \\ \ \\ \ \\ \ } &
\compd{ \ \\ \textcircled b \\ \textcircled c \\ \textcircled d \\ \zero } &
\compd{ \ \\ \textcircled b \\ \ \\ \ \\ \zero } &
\compd{ \ \\ b \\ \textcircled c \\ \textcircled d \\ \textcircled \zero } &
\compd{ \ \\ \textcircled b \\ \ \\ \ \\ \zero } &
\compd{ a \\ \ \\ \ \\ \ \\ \ } &
\compd{ \ \\ \textcircled b \\ c \\ d \\ \zero } &
\compd{ \ \\ \textcircled b \\ \ \\ \ \\ \zero } &
\compd{ \ \\ b \\ \textcircled c \\ \textcircled d \\ \textcircled \zero } &
\compd{ \ \\ \textcircled b \\ \ \\ \ \\ \zero } &
\compd{ \ \\ \ \\ \textcircled c \\ \textcircled d \\ \zero } \\
\end{tabular}
\end{center}
and render the second constraint unsatisfiable.
The same course of reasoning holds for assigning $b$ to $X_2$ and $X_5$.
Thus only possibility is to assign $b$ to $X_3$,
which leaves only two places free in the period of the second \SPACINGONE,
hence two symmetrical solutions.

Formally, we first need to count a number of places that may contain a value from $S_l$
for each \SPACINGONE\ constraint $l$. We denote it by $u_l$ and define as
$$
u_l = |\setcomp{i}{S_l \cap P^l_i \neq \emptyset, 1 \leq i \leq p_l}|
$$
where $P^l_i$ are the folded domains from the proof of Theorem~\ref{thm:tract1} for the constraint $l$.
Now we need to count a number of places that are not blocked for values from $S_{l_2}$,
but would become blocked after assigning a value from $S_{l_1}$ to a variable
$X_i$ in the first period of $l_1$ (i.e. $1 \leq i \leq p_{l_1}$).
This count is denoted by $b^{l_1, l_2}_i$ and defined as
$$
% b^{l_1, l_2}_i = |\setcomp{x}{x = ((i + j p_{l_1} -1) \mymod p_{l_2})+1, 0 \leq j < k_{l_1}, i + j p_{l_1} \leq k_{l_2} p_{l_2}, S_{l_2} \cap D(X_x) \neq \emptyset}|
|\setcomp{x}{x = ((i + j p_{l_1} -1) \mymod p_{l_2})+1, 0 \leq j < k_{l_1}, i + j p_{l_1} \leq k_{l_2} p_{l_2}, S_{l_2} \cap D(X_x) \neq \emptyset}|
$$
for each pair of different voices $l_1$, $l_2$ and $1 \leq i \leq p_{l_1}$.

\begin{theorem}\label{thm:soundR}
Let $l_1 = \SPACINGONE(S_{l_1}, p_{l_1}, k_{l_1}, X)$ and $l_2 = \SPACINGONE(S_{l_2}, p_{l_2}, k_{l_2}, X)$
be two constraints defined on the same sequence $X$ with $S_{l_1} \cap S_{l_2} = \emptyset$ that are \DC.
Let $u_{l_2}$ and $b^{l_1, l_2}_i$ be defined as above.
If $b^{l_1, l_2}_i > u_{l_2} - |S_{l_2}|$ then
there is no support for $l_1 \  \wedge \  l_2$ with any value from $S_{l_1}$ assigned to $X_i$
within $1 \leq i \leq p_{l_1}$.
\end{theorem}

\begin{proof}
Suppose the premises of the theorem hold.
The places in the first period of $l_2$ that
may hold a value from $S_{l_2}$ (together with places that already hold it) are called \emph{free}.
The rest of the places in the first period of $l_2$ are called \emph{blocked}.
The number of the free places is $u_{l_2}$.
Suppose that we assigned a value $d_1 \in S_{l_1}$ to $X_i$.
Then $d_1$ has to repeat on the same place in each period of $l_1$,
that is places $i + j p_{l_1}$ for each $0 \leq j < k_{l_1}$.
Now no value from $S_{l_2}$ can be put on those of these places
that are constrained by $l_2$, i.e. $i + j p_{l_1} \leq k_{l_2} p_{l_2}$ for $0 \leq j < k_{l_1}$.
So no value from $S_{l_2}$ can repeat on whatever places in the period $p_{l_2}$ these places are aligned with.
Arbitrary index $y$ into the sequence $X$ is $((y-1) \mymod p_{l_2})+1$-st place in a period of $l_2$,
so no value from $S_{l_2}$ can be assigned to any variable with index $x = ((i + j p_{l_1} -1) \mymod p_{l_2})+1$
where $i + j p_{l_1} \leq k_{l_2} p_{l_2}$ for $0 \leq j < k_{l_1}$.
While $l_2$ is \DC\ and $1 \leq x \leq p_{l_2}$, $S_{l_2} \cap D(X_x) = S_{l_2} \cap P^{l_2}_x$.
This restricts the places we are counting in $b^{l_1, l_2}_i$ to those that were free,
so $b^{l_1, l_2}_i$ is the number of places that are blocked exclusively by assigning $d_1$ to $X_i$.
Finally, suppose $b^{l_1, l_2}_i > u_{l_2} - |S_{l_2}|$ holds.
This is equivalent to $u_{l_2} - b^{l_1, l_2}_i < |S_{l_2}|$ which means that
the number of free places without the places blocked solely by $d_1$ on the place $i$
(the places left free after the assignment)
is less than the number of values from $S_{l_2}$ we need to assign to one period of $l_2$.
The constraint $l_2$ is obviously unsatisfiable in this case,
so there is no support for $l_1 \  \wedge \  l_2$ with any value $d \in S_{l_1}$ assigned to $X_i$.
\qed
\end{proof}

\section{Application}

One of the applications of CSP in music composition is
the problem of \emph{Asynchronous rhythms} described by Truchet~\cite{TruchetPhD,MusicBook}.
Having $h$ voices and a time horizon $H$,
the goal is to construct one rhythmical pattern for each voice.
The patterns are played repeatedly until the time horizon is reached.
The pattern of voice $l$ is of length $p_l$ and consists of $m_l$ onsets that
have to be placed so that no two onsets are played at the same time.
We consider also extension of this problem, where
composer may want to impose additional constraints, such as
forbidding or enforcing particular onsets on particular places, or
restricting the density of the onsets e.g. at most one onset on two successive places.

Truchet~\cite{TruchetPhD,MusicBook} formalized this problem in terms of variables $V_{l,i}$ that denote time
when an onset $i$ is played in the pattern of voice $l$ (so $D(V_{l,i}) = \set{1, \ldots, p_l}$) and
the following constraints: $\AllDifferent(\set{V_{l,1}, \ldots, V_{l,m_l}})$ for each voice $l$, and
$V_{l_1,i_1} + j_1 p_{l_1} \neq V_{l_2,i_2} + j_2 p_{l_2}$ for each pair of different voices $l_1$, $l_2$,
each onset $i_1$, $i_2$ in each voice respectively and
each $j_1$, $j_2$ s.t. $V_{l_1,i_1} + j_1 p_1 \leq H$ and $V_{l_2,i_2} + j_2 p_2 \leq H$.
This, however, cannot be encoded as a CSP model,
because if $H \mymod p_l \neq 0$ for some voice $l$ then
the range of $j_l$ in the later constraint depends on $V_{l,i}$ for each onset $i$.
We fix this by minor modification, where
the pattern of each voice $l$ is repeated exactly $k_l$ times and
the indexes $j_1$, $j_2$ in the last constraint are quantified by $0 \leq j_1 < k_{l_1}$ and $0 \leq j_2 < k_{l_2}$.
We will call this model \oModel.

We encoded the problem into three new models.
The first model (\sModel)
uses one $\SPACINGONE(S_l, p_l, k_l, X)$ for each voice $l$ on the same sequence of variables $X$.
The sets of values $S_l$ represent the onsets for each voice $l$, so they must be pairwise disjoint.
The length of the pattern of voice $l$ is $p_l$ and the number of its repetitions is $k_l$.
Each variable $X_j$ represents the onset played in the beat $j$, so
$D(X_j) = \bigcup_{l = 1}^h S_l \cup \set{\zero}$ where $\zero \not\in \bigcup_{l = 1}^h S_l$.

\begin{theorem}\label{thm:strongS}
Enforcing \DC\ on \sModel\ is strictly stronger than enforcing \DC\ on \oModel.
\end{theorem}

\begin{proof}
First we show that
if every constraint in \sModel\ is \DC, then also every constraint in \oModel\ is \DC.

Suppose every constraint in \sModel\ is \DC.

The first question is what is the correspondence between the two models in this case,
i.e. if \sModel\ is \DC, with what domains should we check \oModel\ for \DC.
The correspondence should be semantical. That is, for every placing of every onset whether
it is possible to place other onsets so that each constraint of a model is satisfied.
While all the \SPACINGONE\ constraints are \DC, the domains restricted to their values
are the same in each of their period.
So the correspondence between the models can be described only on the first periods.
It is $d \in D(X_i)$ where $d \in S_l$ and $1 \leq i \leq p_l$ iff $i \in D(V_{l,d})$.

If every constraint in \sModel\ is \DC, then also every \AllDifferent\ in \oModel\ is \DC,
because the bipartite graph in the propagator of \SPACINGONE\ of voice $l$
is a supergraph of \AllDifferent\ for voice $l$ in \oModel.
So while the \SPACINGONE\ is \DC, all edges in its graph belong to some maximum matching.
If we remove each \zero\ node, we obtain a graph of the corresponding \AllDifferent.
Each of the edges in the new graph still belong to some maximum matching in it,
because removing nodes preserves maximum matchings, because
while each edge is in some maximum matching and
each node is in at most one edge of a matching,
each maximum matching loses exactly one edge per removed node.

If every constraint in \sModel\ is \DC, then also every difference constraint in \oModel\ is \DC.
The semantics of a difference constraint $V_{l_1,i_1} + j_1 p_{l_1} \neq V_{l_2,i_2} + j_2 p_{l_2}$ is that
the onset $i_1$ in the $j_1+1$-st period of the voice $l_1$ cannot be played at the same time as
the onset $i_2$ in the $j_2+1$-st period of the voice $l_2$.
If a difference constraint was not \DC, this would mean that for some placing of the
onset $i_1$ it is not possible to place the onset $i_2$ so that they would be played
on different times in the $j_1+1$-st and $j_2+1$-st periods of their voices respectively.
This is possible only if the placing of $i_2$ is fixed ($D(V_{l_2,i_2}) = \set{x}$).
This, however, means in the \sModel\ that the value of this onset has
only one place where it can occur in one period of $l_2$
($i_2 \in D(X_y)$ iff $y = x$ for $1\leq y \leq p_{l_2}$),
so the propagator of \SPACINGONE\ for this voice would instantiate it and
hence remove other values from the domain of the variable of the place ($D(X_x) = \set{i_2}$).
This means for \SPACINGONE\ of $l_1$ instantiating the variable of the place to \zero\ ($D(Y^{l_1}_x) = \set{\zero}$),
so assuming \sModel\ is \DC\ we have contradiction with possibility of placing $i_1$
so that the difference constraint cannot be satisfied
($x \notin D(V_{l_1,i_1})$ because $i_1 \notin D(Y^{l_1}_x) = \set{\zero}$),
so the difference constraint must be \DC.

To show strictness, consider instance of \emph{Asynchronous rhythms} with two voices,
first having $m_1 = p_1 = k_1 = 2$ and the second one $m_2 = 1$, $p_2 = 3$, $k_2 = 2$.
Additionally, the very first place has to contain an onset of the first voice.
Domains of variables of \oModel\ for this instance are
$D(V_{1,a}) = D(V_{1,b}) = \set{1,2}$ and $D(V_{2,c}) = \set{2,3}$
and the constraints are $\AllDifferent(\set{V_{1,a}, V_{1,b}})$ and
$V_{1,i} \neq V_{2,c}$,
$V_{1,i} + 2 \neq V_{2,c}$,
$V_{1,i} \neq V_{2,c} + 3$,
$V_{1,i} + 2 \neq V_{2,c} + 3$ for $i \in \set{a,b}$.
All the constraints of this model are \DC.
However, domains of variables of \sModel\ are
$D(X_1) = \set{a,b}$ and $D(X_2) = \ldots = D(X_6) = \set{a,b,c,\zero}$ where
$S_1 = \set{a,b}$ and $S_2 = \set{c}$.
The propagation of \SPACINGONE\ for the first voice will remove all values not in $S_1$ from $D(X_2)$ up to $D(X_4)$.
This leaves the \SPACINGONE\ for the second voice unsatisfiable.
So \sModel\ determines that the instance is unsatisfiable without search while \oModel\ is not able to do so.
\qed
\end{proof}

Drawback of \sModel\ is that,
due to interchangeability of onsets in one voice, it generates a lot of symmetrical solutions.
This observation led to the second model (\bsModel).
It is the same as the first one, except it uses a version of \SPACINGONE\ that
does not distinguish between different values in one voice,
which is a piece of information a composer does not need when constructing rhythmical patterns.
This new constraint, denoted $\BLINDSPACING(d, m, p, k, X)$, can be defined as $\SPACINGONE(S, p, k, X)$
where $S$ is a multiset of $m$ values $d$.
In other words, $\BLINDSPACING(d, m, p, k, X)$ is satisfied iff
there is exactly $m$ occurrences of $d$ in the first $p$ places of $X$
and this pattern is repeated $k$ times.
Filtering algorithm of \BLINDSPACING\ takes $O(n)$ time down a branch of a search tree
as we can avoid finding the maximum matching step.
After folding of the domains $P_i = \bigcap_{j=0}^{k-1} D(X_{j p + i})$ for $1 \leq i \leq p$,
the filtering algorithm simply counts $u = |\setcomp{i}{d \in P_i, 1 \leq i \leq p}|$
and $v = |\setcomp{i}{P_i = \set{d}, 1 \leq i \leq p}|$.
The algorithm fails iff $u < m$ or $v > m$,
if $u = m$ it instantiates all $X_i$ with $d \in D(X_i)$ to $d$ and
if $v = m$ it removes $d$ from all domains that $D(X_i) \neq \set{d}$.
Of course, the algorithm reflects the folded domains $P_i$ back to the domains $D(X_i)$
by removing $d$ if $d \notin P_i$ or instantiating to $d$ if $P_i = \set{d}$
and repeats all the instantiations and all the removals on the appropriate places in each period.

\begin{theorem}\label{thm:correctBS}
% The filtering algorithm for \BLINDSPACING\ fails iff the constraint is unsatisfiable.
% If it does not fail, it establishes \DC.
Enforcing \DC\ on \BLINDSPACING\ can be done in $O(n)$ time down a branch of the search tree.
\end{theorem}

\begin{proof}
In order to show that the filtering algorithm described above really establishes \DC,
we need to show that it fails iff the constraint is unsatisfiable and
if it does not fail the constraint is \DC.

Thanks to the folding of the domains, reflecting the folded domains back and
repetition of instantiations and removals in each period,
for the domains on the same places of each period ($\setcomp{D(X_{jp+i})}{0 \leq j < k}$ for each $i$)
it holds that one contain $d$ or is $\set{d}$ iff all the others contain $d$ or are $\set{d}$ respectively.
This allows us to restrict out reasoning only to one period.

First, we show that the algorithm fails iff the constraint is unsatisfiable.
After what is shown in the previous paragraph the only way how to falsify the constraint is
to put either too many or too little values $d$ to one period.
Obviously, $u < m$ means that there is not enough variables in one period with $d$ in their domains
and $v > m$ means that $d$ is assigned to too many variables in one period.
So $u < m$ or $v > m$ holds iff the constraint is unsatisfiable,
so the algorithm fails iff the constraints is unsatisfiable.

Second, we show that if the algorithm removes a value from the domain of a variable,
there is no solution of the constraint with this value assigned to this variable.
Suppose the algorithm removes $d$ from $D(X_i)$.
Except the cases considered in the first paragraph,
the algorithm does so only when $v = m$ and $i$ is not one of the places counted in $v$.
So if there was a solution with $X_i = d$,
the number of occurrences of $d$ in one period would be $m+1$, which is a contradiction.
Suppose the algorithm removes $d' \neq d$ from $D(X_i)$.
Except the cases considered in the first paragraph,
the algorithm does so only when $u = m$ and $i$ is one of the places counted in $u$.
So if there was a solution with $X_i = d'$,
the number of occurrences of $d$ in one period would be $m-1$, which is a contradiction.

Finally, we show that if the constraint is satisfiable and
there is no solution of the constraint with a value assigned to a variable,
the algorithm removes the value from the domain of the variable.
Suppose the constraint is satisfiable and
there is no solution of the constraint with a value $d'$ assigned to a variable $X_i$.
Except the cases considered in the first paragraph,
the constraint may have no solution only in two cases:
The value $d$ is assigned to too many or too little variables in one period.
The constraint is satisfiable without the assignment $X_i = d'$.
So the first case is possible only if there are exactly $m$ variables in one period
with $d$ assigned to them and the assignment $X_i = d'$ increases this number.
This means that $X_i$ is not instantiated to $d$ yet, $d' = d$ and $v = m$ holds,
so the algorithm will remove $d$ from the domains of all the variables that
are not instantiated to $d$ yet, hence also from $D(X_i)$.
The second case is possible only if there are exactly $m$ variables in one period
with $d$ in their domains and the assignment $X_i = d'$ decreases this number.
This means that $X_i$ contains $d$ in its domain, $d' \neq d$ and $u = m$ holds,
so the algorithm will assign $d$ to all the variables that contain $d$ in their domains,
hence remove $d'$ from $D(X_i)$.

The algorithm runs in $O(n)$ time down a branch of the search tree, because
the folded domains $P_i$ are represented implicitly.
At the beginning of the search, folding and reflecting back to the domains takes $O(pk)$,
because it is done simply by, for every $1 \leq i \leq p$,
removing $d$ from each $D(X_{jp+i})$, $0 \leq j < k$ if it is missing in some and
instantiating each $X_{jp+i}$, $0 \leq j < k$ to $d$ if some already is instantiated.
During the search it is sufficient just to
update $u$ and $v$ upon removal of $d$ or instantiation to $d$.
On each update, the same action (removal or instantiation) has to be repeated in each of $k$ period.
Down a branch there will be at most $p$ such updates, because after $p$ updates
the pattern of $d$-s in a period is fully determined.
That makes $O(pk)$ steps down a branch.
While $pk \leq n$, the filtering runs in $O(n)$ time down a branch of the search tree.
\qed
\end{proof}

The last model (\srModel) is the same as \sModel,
but it is additionally making use of the incomplete filtering rule for \SPACINGH\ 
(described in the previous section) between each pair of constraints.

\section{Experimental results}

To compare performance of the different models,
we carried out a series of experiments on random instances of the \emph{Asynchronous rhythms} problem.
The instances were generated for a fixed number of voices $h$,
a mean length of the pattern of the first voice $p_1$ and
a fixed number of repetitions of the last voice $k_h$.
The models were tested on ten instances for each tuple of $(h, p_1, k_h)$.
The generation of the instances followed the philosophy that
the first voice should be the base voice with short pattern and small number of onsets
and the other voices should have richer patterns. %more complicated.
So the pattern of the second voice was $4 \pm 1$ beats longer that the pattern of the first voice and
each following voice had pattern in average two times longer $\pm 3$ beats than two voices before.
Numbers of repetitions were set so that the voices overlap as much as possible.
And the total number of the onsets was approximately $75\%$ of the length of the sequence
uniformly distributed between the voices.
The experiments were run with 5 minute timeout and with a heuristic under which all models were performing better.
The left side of Table~\ref{t:t1} shows performance of the models on the \emph{basic} Asynchronous rhythms problem.
The right side of Table~\ref{t:t1} shows performance of the models on
the \emph{extended} Asynchronous Rhythms problem, where the composer applies 
 additional constraints that some onsets are forbidden in certain positions.
We randomly removed $10\%$ of values from the domains to model this restriction.
\bsModel\ was not tested on the extended problem, because in this case onsets are not interchangeable due 
to additional constraints on onsets.
Experiments were run with CHOCO Solver 2.1.5 on  Intel Xeon 3 CPU 2.0Ghz, 4GB RAM.

\begin{table}
\begin{center}
{\scriptsize
\caption{
	\label{t:t1}Number of solved instances / average time to solve in sec / average number of backtracks in thousands.
}
\begin{tabular}{|ccc||crr|crr|crr|crr||crr|crr|crr|}
   \hline
  & & & \multicolumn{12}{|c||}{Basic problem} & \multicolumn{9}{c|}{Extended problem} \\
  $h$ & $p_1$ & $k_h$ & \multicolumn{3}{|c}{\oModel} & \multicolumn{3}{c}{\sModel} & \multicolumn{3}{c}{\srModel} & \multicolumn{3}{c||}{\bsModel} & \multicolumn{3}{c}{\oModel} & \multicolumn{3}{c}{\sModel} & \multicolumn{3}{c|}{\srModel} \\
  \hline
  3 & 12 & 2 & 9 & 4.99 & 99.76 & 9 & 9.54 & 61.03 & 9 & 4.83 & 24.38 & \textbf{10} & \textbf{0.17} & \textbf{4.58} & 8 & \textbf{0.72} & \textbf{4.29} & \textbf{10} & 21.85 & 143.66 & \textbf{10} & 9.82 & 42.08 \\
  3 & 12 & 3 & \textbf{10} & 0.84 & 2.80 & \textbf{10} & 1.13 & 1.80 & \textbf{10} & 0.60 & 0.46 & \textbf{10} & \textbf{0.09} & \textbf{0.11} & 9 & 8.56 & 180.02 & \textbf{10} & 0.17 & 0.02 & \textbf{10} & \textbf{0.12} & \textbf{0.00} \\
  3 & 12 & 4 & \textbf{10} & 3.92 & 35.17 & \textbf{10} & 4.86 & 24.88 & \textbf{10} & 1.50 & 2.92 & \textbf{10} & \textbf{0.18} & \textbf{0.33} & \textbf{10} & 0.57 & 0.98 & \textbf{10} & 0.31 & 0.04 & \textbf{10} & \textbf{0.20} & \textbf{0.02} \\
  3 & 18 & 2 & 7 & 21.02 & 244.97 & 7 & 20.83 & 102.16 & 7 & 9.86 & 35.60 & \textbf{10} & \textbf{0.64} & \textbf{18.05} & 5 & 7.22 & 138.13 & \textbf{9} & 29.60 & 109.82 & \textbf{9} & \textbf{5.96} & \textbf{13.47} \\
  3 & 18 & 3 & 8 & 13.48 & 128.47 & 8 & 16.56 & 58.54 & 8 & 7.58 & \textbf{22.13} & \textbf{10} & \textbf{1.01} & 26.95 & 6 & 10.38 & 71.18 & 8 & \textbf{0.25} & \textbf{0.03} & \textbf{9} & 6.86 & 14.60 \\
  3 & 18 & 4 & 7 & 42.77 & 270.29 & 7 & 44.88 & 119.39 & 7 & 9.25 & \textbf{14.21} & \textbf{10} & \textbf{1.05} & 17.42 & \textbf{10} & 31.71 & 137.36 & \textbf{10} & 0.44 & 0.30 & \textbf{10} & \textbf{0.34} & \textbf{0.13} \\
  3 & 24 & 2 & 6 & 5.70 & 48.93 & 7 & 11.39 & 29.52 & 7 & \textbf{5.56} & \textbf{8.31} & \textbf{10} & 20.80 & 1409.94 & 3 & 10.37 & 176.43 & 7 & \textbf{3.31} & \textbf{10.80} & \textbf{8} & 37.01 & 63.47 \\
  3 & 24 & 3 & 3 & 1.58 & 6.79 & 3 & 1.20 & 2.62 & 3 & \textbf{0.46} & \textbf{0.34} & \textbf{10} & 4.49 & 90.87 & 2 & 121.36 & 286.81 & \textbf{9} & 47.57 & 166.14 & \textbf{9} & \textbf{9.58} & \textbf{21.07} \\
  3 & 24 & 4 & 2 & 8.38 & \textbf{40.97} & 2 & 22.11 & 41.56 & 4 & 86.38 & 59.74 & \textbf{10} & \textbf{5.10} & 91.03 & 4 & 51.02 & 195.04 & \textbf{10} & 46.92 & 114.67 & \textbf{10} & \textbf{5.71} & \textbf{9.36} \\
  \hline
  4 & 12 & 2 & 9 & 3.78 & 63.49 & 9 & 6.37 & 35.25 & 9 & 2.48 & \textbf{8.06} & \textbf{10} & \textbf{0.43} & 10.00 & 7 & \textbf{3.15} & 46.06 & \textbf{10} & 8.62 & 20.04 & \textbf{10} & 4.69 & \textbf{7.74} \\
  4 & 12 & 3 & 8 & 26.01 & 523.55 & 7 & 0.90 & 1.19 & 7 & 0.64 & \textbf{0.41} & \textbf{10} & \textbf{0.36} & 4.90 & \textbf{10} & 39.23 & 450.50 & \textbf{10} & 4.76 & 15.61 & \textbf{10} & \textbf{3.06} & \textbf{8.95} \\
  4 & 12 & 4 & 8 & 21.81 & 546.20 & 8 & 31.88 & 208.21 & 9 & 38.60 & 63.99 & \textbf{10} & \textbf{0.39} & \textbf{4.54} & 9 & \textbf{1.70} & 14.58 & \textbf{10} & 9.59 & 36.59 & \textbf{10} & 3.91 & \textbf{9.87} \\
  4 & 18 & 2 & 8 & 29.46 & 438.73 & 8 & 26.11 & 88.45 & 8 & 17.56 & 50.75 & \textbf{10} & \textbf{0.61} & \textbf{20.28} & 5 & 62.00 & 837.42 & \textbf{8} & 3.13 & 10.41 & \textbf{8} & \textbf{1.30} & \textbf{2.52} \\
  4 & 18 & 3 & 6 & 1.40 & 12.64 & 6 & 1.95 & 5.02 & 6 & \textbf{1.04} & \textbf{1.44} & \textbf{10} & 1.41 & 20.64 & 7 & 6.89 & 20.61 & \textbf{10} & 13.20 & 17.57 & \textbf{10} & \textbf{6.75} & \textbf{8.35} \\
  4 & 18 & 4 & 5 & 5.06 & 15.08 & 5 & 10.33 & 11.09 & 5 & \textbf{1.64} & \textbf{0.63} & \textbf{10} & 2.85 & 36.14 & 8 & 19.17 & 144.31 & \textbf{10} & 1.43 & 1.35 & \textbf{10} & \textbf{0.78} & \textbf{0.39} \\
  4 & 24 & 2 & 5 & \textbf{0.14} & 0.01 & 5 & 0.30 & \textbf{0.00} & 6 & 37.27 & 37.70 & \textbf{8} & 2.06 & 54.92 & 3 & 51.49 & 501.10 & 7 & \textbf{6.69} & \textbf{18.52} & \textbf{8} & 25.07 & 38.43 \\
  4 & 24 & 3 & 4 & 9.41 & 23.10 & 4 & 14.15 & 15.90 & 5 & \textbf{7.43} & \textbf{4.34} & \textbf{9} & 24.67 & 507.45 & 2 & \textbf{1.41} & \textbf{0.44} & \textbf{8} & 17.49 & 23.72 & \textbf{8} & 4.05 & 4.50 \\
  4 & 24 & 4 & 1 & 1.35 & 1.24 & 1 & 1.41 & 0.84 & 1 & \textbf{0.65} & \textbf{0.17} & \textbf{10} & 8.39 & 125.89 & 4 & 27.64 & 101.35 & \textbf{9} & 5.38 & 10.56 & \textbf{9} & \textbf{3.16} & \textbf{4.93} \\
  \hline
  5 & 12 & 2 & \textbf{10} & 0.98 & 6.69 & \textbf{10} & 0.71 & 0.52 & \textbf{10} & 0.53 & 0.32 & \textbf{10} & \textbf{0.08} & \textbf{0.07} & 6 & 1.30 & 11.45 & \textbf{10} & \textbf{0.26} & 0.05 & \textbf{10} & 0.28 & \textbf{0.03} \\
  5 & 12 & 3 & 7 & \textbf{0.07} & 0.02 & 7 & 0.37 & \textbf{0.01} & 9 & 26.66 & 31.66 & \textbf{10} & 0.50 & 4.82 & 9 & 3.96 & 57.02 & \textbf{10} & 0.31 & 0.04 & \textbf{10} & \textbf{0.27} & \textbf{0.03} \\
  5 & 12 & 4 & \textbf{10} & 1.40 & 2.94 & \textbf{10} & 5.18 & 4.79 & \textbf{10} & 1.05 & \textbf{0.47} & \textbf{10} & \textbf{0.32} & 0.50 & \textbf{10} & 13.77 & 134.70 & \textbf{10} & 0.13 & 0.00 & \textbf{10} & \textbf{0.13} & \textbf{0.00} \\
  5 & 18 & 2 & 6 & \textbf{0.43} & 0.18 & 7 & 1.39 & 0.42 & 7 & 0.65 & \textbf{0.12} & \textbf{10} & 0.51 & 5.24 & 6 & 20.51 & 80.34 & 9 & \textbf{5.55} & \textbf{9.45} & \textbf{10} & 25.57 & 22.88 \\
  5 & 18 & 3 & 7 & 6.52 & 6.77 & 7 & 6.61 & 2.36 & 7 & 3.64 & \textbf{1.04} & \textbf{10} & \textbf{0.70} & 12.51 & 7 & \textbf{0.93} & \textbf{7.00} & 9 & 14.12 & 38.74 & \textbf{10} & 16.38 & 31.21 \\
  5 & 18 & 4 & 7 & 16.95 & 14.71 & 7 & 38.33 & 17.58 & 7 & \textbf{2.04} & \textbf{0.53} & \textbf{10} & 6.74 & 83.41 & \textbf{10} & 0.30 & 0.12 & \textbf{10} & \textbf{0.18} & 0.00 & \textbf{10} & 0.18 & \textbf{0.00} \\
  5 & 24 & 2 & 5 & 38.12 & 31.28 & 4 & \textbf{2.56} & \textbf{1.24} & 5 & 8.11 & 1.77 & \textbf{10} & 29.43 & 841.65 & 1 & 0.82 & 0.10 & \textbf{8} & 0.45 & 0.04 & \textbf{8} & \textbf{0.45} & \textbf{0.01} \\
  5 & 24 & 3 & 6 & 34.80 & 18.90 & 5 & \textbf{0.66} & \textbf{0.06} & 6 & 8.71 & 0.92 & \textbf{9} & 4.26 & 35.88 & 4 & 54.11 & 76.50 & \textbf{10} & 16.73 & 15.73 & \textbf{10} & \textbf{6.08} & \textbf{3.64} \\
  5 & 24 & 4 & 3 & \textbf{19.47} & 106.09 & 4 & 33.87 & 47.46 & 5 & 29.96 & \textbf{27.82} & \textbf{8} & 21.05 & 124.15 & 6 & 4.66 & 10.80 & \textbf{10} & 1.02 & 0.46 & \textbf{10} & \textbf{0.82} & \textbf{0.36} \\
  \hline
\end{tabular}
}
\end{center}
\end{table}

As we can see in the results, \bsModel\ solved almost all instances and,
where comparable, it was the fastest and needed the least backtracks
in solving basic model. \bsModel\ is so successful because it removes symmetries and
the filtering algorithm of \BLINDSPACING\ runs in $O(n)$ down a branch.
On the basic problem, \sModel\ is not obviously better than \oModel,
however \sModel\ performs much better on the extended problem and
it needs significantly less backtracks than \oModel.
That supports the theory that \sModel\ achieves more propagation than \oModel.
Finally, the additional rule significantly improves performance of \srModel\ against \sModel,
especially the number of backtracks is lower by order of magnitude, which
shows that the rule really facilitates propagation between the \SPACING\ constraints.

\section{Conclusions}

The global $\SPACING$ constraint is useful in modeling events that are distributed over time,
like learning units scheduled over a study program or repeated patterns in music compositions.
We have investigated theoretical properties of the constraint and
shown that enforcing domain consistency (\DC) is intractable even in very restricted cases.
On the other hand we have identified two tractable restrictions
and implemented efficient \DC\ filtering algorithms for one of them.
The algorithm takes $O(p^2 k + p^{2.5})$ time down a branch of the search tree.
We have also proposed an incomplete filtering algorithm for one of the intractable cases.
We have experimentally evaluated performance of the algorithms on a music composition
problem and demonstrated that our filtering algorithms outperform the 
state-of-the-art approach for solving this problem in both, speed and number of backtracks.

% \bibliographystyle{plain}
% \bibliography{lit}
% %\bibliography{/Users/twalsh/Documents/biblio/a-z2.bib,/Users/twalsh/Documents/biblio/pub2.bib}

\end{document}